\documentclass[journal]{IEEEtran}
\usepackage{cite}
\usepackage{longtable}
\usepackage{algorithmic}
\usepackage{graphicx}
\usepackage{caption}
\usepackage{textcomp}
\usepackage{enumerate}
\usepackage{latexsym}
\usepackage{amsmath}
\usepackage{amsthm}
\usepackage{amssymb}
\usepackage{epsfig}
\usepackage{overpic}
\usepackage{multirow}
\usepackage[table]{xcolor}
\usepackage{colortbl,array}
\usepackage{multirow,bigdelim}
\usepackage{hyperref}
\usepackage{subfigure}
\usepackage{booktabs}
\usepackage{tabularx}
\usepackage{fancyhdr}
\usepackage{nomencl}
\usepackage[flushleft]{threeparttable}
\usepackage{figsize}
\usepackage{mathtools}
\usepackage{algorithm}
\usepackage{soul}
\usepackage[utf8]{inputenc}
\usepackage[english]{babel}
\usepackage{arydshln}
\usepackage{tikz}
\DeclareMathOperator{\diag}{diag}
\DeclareMathOperator{\blkdiag}{blkdiag}
\DeclareMathOperator{\sign}{sign}
\newtheorem{thm}{Theorem}

\newtheorem{rem}{Remark}
\newtheorem{cor}{Corollary}
\newtheorem{asm}{Assumption}

\newtheorem{ex}{Example}

\def\BibTeX{{\rm B\kern-.05em{\sc i\kern-.025em b}\kern-.08em
    T\kern-.1667em\lower.7ex\hbox{E}\kern-.125emX}}
\begin{document}
\title{Multidimensional Opinion Dynamics with Confirmation Bias: A Multi-Layer Framework}
\author{M.Hossein~Abedinzadeh, \IEEEmembership{Student Member, IEEE} and~Emrah Akyol, \IEEEmembership{Senior Member, IEEE}
\thanks{M.Hossein~Abedinzadeh and E.~Akyol are with the Department of Electrical and Computer Engineering, Binghamton University--SUNY. e-mail: \{mabedin3, eakyol\}@binghamton.edu.}
\thanks{This research is supported by the NSF via  (CAREER) CCF \#2048042.}
\thanks{This paper was presented in part in Allerton Conference 2023.}
}

\maketitle

\begin{abstract}

We study multidimensional opinion dynamics under confirmation bias in social networks. Each agent holds a vector of correlated opinions across multiple topic layers. Peer interaction is modeled through a static, informationally symmetric social channel, while external information enters through a dynamic, informationally asymmetric source channel. Source influence is described by nonnegative state-dependent functions of agent--source opinion mismatch, which captures confirmation bias without hard thresholds. For general Lipschitz source-influence functions, we give sufficient conditions under which the dynamics are contractive and converge to a unique steady state independent of the initial condition. For affine confirmation-bias functions, we show that the steady state can be computed through a finite sign-consistency search and identify a regime in which it admits a closed form. For broader classes of bounded nonlinear source-influence functions, we derive explicit lower and upper bounds on the fixed point. Numerical examples and a study on a real-world adolescent lifestyle network illustrate the role of multidimensional coupling and show that source-design conclusions can change qualitatively when confirmation bias is ignored.

\end{abstract}

\section{Introduction}
\label{sec:introduction}

Empirical studies in the social sciences suggest that individuals often form opinions on different issues in a correlated fashion; opinion states therefore evolve as vectors rather than as isolated scalars \cite{dellaposta2015liberals,block2014multidimensional,baumann2021emergence}. For example, \cite{dellaposta2015liberals} documents a strong association between political ideology and lifestyle choices and shows that even small initial differences can be amplified over time through homophily and social influence. While such phenomena are well documented empirically, mathematical models that simultaneously capture multidimensional opinions, inter-topic coupling, and selective exposure to information remain limited.

A key mechanism behind selective exposure is \emph{confirmation bias}, namely the tendency to seek, interpret, and weigh information in ways that reinforce pre-existing beliefs \cite{nickerson1998confirmation}. In networked information environments, confirmation bias can amplify polarization and facilitate the spread of misinformation \cite{del2016spreading,del2017modeling,del2018polarization}. The model studied in this paper is intended as a tractable building block for such settings, where agents interact both with peers and with external information sources \cite{8510828,10313365}.

Classical models of opinion dynamics include the DeGroot model \cite{degroot1974reaching}, the Friedkin--Johnsen (FJ) model \cite{friedkin1990social}, and the Hegselmann--Krause (HK) bounded-confidence model \cite{hegselmann2002opinion}. In DeGroot dynamics, each agent updates its opinion by averaging the opinions of its neighbors, which often leads to consensus. The FJ model addresses part of this limitation by incorporating innate opinions. The HK model takes a different route: agents interact only when their opinions are sufficiently close. This state dependence captures selective interaction, but it also introduces discontinuities that make analysis substantially more delicate. In the scalar homogeneous HK model, these discontinuities still permit finite-time convergence results \cite{lorenz2006consensus,blondel2009krause}, but they do not yield a general analytical characterization of the steady state. Moreover, the standard HK framework does not distinguish between ordinary agents and exogenous information sources.

That distinction is practically significant. Interactions among peers are typically informationally symmetric and relatively static: friends, neighbors, and coworkers influence one another through a social graph that is not highly sensitive to short-term opinion fluctuations. By contrast, interaction with information sources such as media outlets, thought leaders, or podcasts is informationally asymmetric, and its effective strength is often state dependent because individuals are more likely to attend to sources that are closer to their current views. The scalar model in \cite{8510828} captures this distinction by combining FJ-type interpersonal updates with confirmation-bias-driven source influence.

The present paper extends that idea to multidimensional opinions. This extension is not straightforward because of the interplay between state-dependent source influence and multiple opinion dimensions.  More broadly, multidimensional models can exhibit cross-topic effects that are invisible in any collection of decoupled scalar systems. Related multidimensional bounded-confidence models have been studied in \cite{de2022multi,parsegov2016novel,etesami2013termination,stamoulas2018convergence,laguna2003vector,fortunato2005vector}. These works focus primarily on HK-type interactions and their convergence properties. In contrast, the present paper studies fixed-point existence, uniqueness, computation, and approximation for a topic-layered FJ-type system with state-dependent confirmation-bias source influence.

From a multilayer-network viewpoint, the model is a restricted node-aligned multiplex system in the sense of \cite{kivela2014multilayer,boccaletti2014structure}: layers index correlated opinion dimensions, the same set of agents and sources appears in every layer, intra-layer social interaction is encoded by $W_l$, and inter-layer coupling is encoded by $\Lambda_{l,j}$. The diagonal structure of $\Lambda_{l,j}$ implies that cross-topic transfer acts within the same agent rather than across different agents. Thus the paper studies a structured multilayer model rather than the most general interlayer architecture.

Our model differs from the multidimensional HK literature in three main ways: i) confirmation bias acts only between agents and information sources, not between pairs of ordinary agents; ii) source influence is modeled through a continuous state-dependent function rather than a hard threshold, which enables tractable analysis in some cases; and iii) innate opinions are incorporated explicitly, as in the FJ model.

Building on the scalar model in \cite{8510828} and our preliminary conference version in \cite{10313365}, the main contributions of this paper are as follows:
\begin{itemize}
\item We formulate a topic-layered multiplex FJ-type model in which each agent holds a vector of correlated opinions, peer influence is coupled across topics, and source influence is state dependent through confirmation-bias weights.
\item For general nonnegative Lipschitz source-influence functions, we give sufficient conditions under which the dynamics are contractive and converge to a unique steady state independent of the initial condition.
\item For affine confirmation-bias functions, we show that the fixed point can be computed through a finite sign-consistency search and identify a regime in which it admits a closed form.
\item For bounded nonlinear source-influence functions, we derive explicit lower and upper bounds on the fixed point.
\item Through numerical examples and a real-world adolescent lifestyle network, we show that multidimensional coupling matters and that source-design conclusions can change qualitatively when confirmation bias is ignored.
\end{itemize}

Relative to \cite{10313365}, the present version adds a complete convergence analysis, a sharper multilayer interpretation, explicit fixed-point bounds, and a real-network source-design study.

The rest of the paper is organized as follows. Section \ref{Preliminaries} introduces notation and the network model. Section \ref{section: Opinion Dynamics Model} presents the multidimensional opinion dynamics. Section \ref{Convergence Analysis} establishes sufficient conditions for convergence to a unique steady state. Section \ref{section: Fixed Point Analysis} studies steady-state computation and approximation. Section \ref{Real Network Simulation} presents numerical experiments on a real-world network, and Section \ref{Conclusion} concludes the paper.
\begin{table*}[t]
\centering
\footnotesize
\setlength{\tabcolsep}{4pt}
\renewcommand{\arraystretch}{1.03}
\caption{Nomenclature}
\begin{tabularx}{\textwidth}{>{\raggedright\arraybackslash}p{0.30\textwidth}|>{\raggedright\arraybackslash}X}
$\mathbb{R}$, $\mathbb{R}_+$, $\mathbb{R}^{m\times n}$ 
& Real numbers, nonnegative real numbers, and $m\times n$ real matrices. \\

$\mathbb{I}$, $\mathbb{I}^n$, $\mathbb{I}^{m\times n}$ 
& The interval $[0,1]$, vectors in $[0,1]^n$, and matrices in $[0,1]^{m\times n}$. \\

$\mathbb{N}$, $\mathbb{N}_0$ 
& Positive integers and nonnegative integers. \\

$\top$, $\sign(\cdot)$ 
& Matrix transposition and sign function. \\

$\mathbf{x}\le \mathbf{y}$, $A\le B$ 
& Componentwise inequalities for vectors and matrices. \\

$\diag(\mathbf{x})$, $\blkdiag(A_1,\dots,A_q)$ 
& Diagonal matrix with diagonal $\mathbf{x}$, and block-diagonal matrix with blocks $A_1,\dots,A_q$. \\

$\Sigma_A$ 
& Diagonal matrix of row sums of $A$, i.e.,
$\Sigma_A \triangleq \diag(A\mathbf{1})$, where $\mathbf{1}$ has compatible dimension. \\

$A\circ B$ 
& Hadamard (entrywise) product of $A$ and $B$. \\

$I$, $\mathbf{0}$, $\mathbf{1}$ 
& Identity matrix, all-zero vector or matrix, and all-one vector, with dimension determined by context. \\

$\mathbf{x}(r)$, $A(r,a)$ 
& The $r$th entry of $\mathbf{x}$, and the $(r,a)$ entry of $A$. \\

$\mathbf{x}[t]$ 
& Value of $\mathbf{x}$ at discrete time $t$. \\

$\|\cdot\|_\infty$ 
& Infinity norm for vectors and induced infinity norm for matrices. \\

$H_i$, $T_k$ 
& Selection matrices that extract the opinion  of agent $i$ and source $k$ from the stacked states $\mathbf{x}$ and $\mathbf{y}$. \\

$\beta_l(i,k)$ 
& Source--agent mismatch score for agent $i$, source $k$, and layer $l$, defined in \eqref{eq:beta}. \\

$b_{l,i,k}(\mathbf{x})$, $B_l(\mathbf{x})$ 
& State-dependent source-influence weight and its matrix form in layer $l$. \\
\end{tabularx}
\label{Nomenclature}
\end{table*}

\section{Preliminaries}
\label{Preliminaries}
Scalars, vectors, and matrices are denoted by lowercase letters, bold lowercase letters, and capital letters, respectively. The remaining notation is summarized in Table~\ref{Nomenclature}.

Let $[q]\triangleq\{1,\dots,q\}$ and $\mathbb{I}\triangleq[0,1]$. The network has agents $\mathcal{V}=\{\mathrm{v}_1,\dots,\mathrm{v}_n\}$ and information sources $\mathcal{U}=\{\mathrm{u}_1,\dots,\mathrm{u}_m\}$. Set $
\mathcal{M}\triangleq \mathcal{V}\cup\mathcal{U}.$

Agent $\mathrm{v}_i$ has expressed opinion $\hat{\mathbf{x}}_i[t]\in\mathbb{I}^q$ and innate opinion $\hat{\mathbf{s}}_i\in\mathbb{I}^q$. The expressed opinion varies in time, whereas the innate opinion is fixed. Information source $\mathrm{u}_k$ has a fixed opinion vector $\hat{\mathbf{y}}_k\in\mathbb{I}^q$.

The multilayer graph is $\mathcal{G}=(\mathcal{M},[q],\mathcal{E})$. An edge $(\mathrm{m}_j,r,\mathrm{v}_i,l)\in\mathcal{E}$, where $\mathrm{m}_j\in\mathcal{M}$, means that node $\mathrm{m}_j$ in layer $r$ influences agent $\mathrm{v}_i$ in layer $l$. This notation is generic. In the dynamics below, peer influence may couple across layers through $\Lambda_{l,j}W_j\mathbf{x}_j[t]$, while source influence acts within each layer through $B_l(\mathbf{x})\mathbf{y}_l$.

For each layer $l\in[q]$, define
\begin{align*}
\mathbf{x}_l[t] &\triangleq [\hat{\mathbf{x}}_1[t](l),\dots,\hat{\mathbf{x}}_n[t](l)]^\top \in \mathbb{I}^n,\\
\mathbf{s}_l &\triangleq [\hat{\mathbf{s}}_1(l),\dots,\hat{\mathbf{s}}_n(l)]^\top \in \mathbb{I}^n,\\
\mathbf{y}_l &\triangleq [\hat{\mathbf{y}}_1(l),\dots,\hat{\mathbf{y}}_m(l)]^\top \in \mathbb{I}^m.
\end{align*}
We then stack the layer-wise vectors as
\begin{align*}
\mathbf{x}[t] &\triangleq [\mathbf{x}_1[t]^\top,\dots,\mathbf{x}_q[t]^\top]^\top \in \mathbb{I}^{nq},\\
\mathbf{s} &\triangleq [\mathbf{s}_1^\top,\dots,\mathbf{s}_q^\top]^\top \in \mathbb{I}^{nq},\\
\mathbf{y} &\triangleq [\mathbf{y}_1^\top,\dots,\mathbf{y}_q^\top]^\top \in \mathbb{I}^{mq}.
\end{align*}

Let $H_i\in\{0,1\}^{q\times nq}$ and $T_k\in\{0,1\}^{q\times mq}$ extract agent and source vectors:
\begin{align*}
\hat{\mathbf{x}}_i[t] = H_i\mathbf{x}[t],
\qquad
\hat{\mathbf{y}}_k = T_k\mathbf{y}.
\end{align*}

\begin{rem}
The vector $\mathbf{x}[t]$ denotes the state at time $t$. When no confusion arises, we suppress the time index.
\label{thm:rmark1}
\end{rem}

\section{Opinion Dynamics Model}
\label{section: Opinion Dynamics Model}
We consider the following model, adapted from the scalar formulation in \cite{8510828}:
\begin{align}
\mathbf{x}_l[t\!+\!1]
\!=\!V_l(\mathbf{x}[t])\mathbf{s}_l
\!+\!A_l\!\sum_{j=1}^{q}\Lambda_{l,j}W_j\mathbf{x}_j[t]
\!+\!A_lB_l(\mathbf{x}[t])\mathbf{y}_l
\label{eq:ABUR}
\end{align}

\begin{enumerate}[(1)]
\item \textit{Social susceptibility matrix} $A_l\in\mathbb{I}^{n\times n}$ is
\begin{align*}
A_l=\diag(\boldsymbol{\alpha}_l),
\end{align*}
where $\boldsymbol{\alpha}_l(i)$, the $i$th component of $\boldsymbol{\alpha}_l\in\mathbb{I}^n$, represents the social susceptibility of agent $\mathrm{v}_i$ in layer $l$.

\item \textit{Cross-topic coupling matrix} $\Lambda_{l,j}\in\mathbb{I}^{n\times n}$ is
\begin{align*}
\Lambda_{l,j}=\diag(\boldsymbol{\lambda}_{l,j}),
\qquad
\sum_{j=1}^{q}\Lambda_{l,j}=I,
\end{align*}
where $\boldsymbol{\lambda}_{l,j}(i)$ specifies how strongly agent $\mathrm{v}_i$ in layer $l$ uses the social aggregate formed in layer $j$.

\item \textit{Weighted adjacency matrix} of layer $l$ is $W_l\in\mathbb{I}^{n\times n}$.

\item \textit{State-dependent source-influence matrix} $B_l(\mathbf{x})\in\mathbb{R}_+^{n\times m}$ represents the influence of the information sources on layer $l$. Its $(i,k)$ entry is denoted by
\begin{align*}
b_{l,i,k}(\mathbf{x})\triangleq B_l(\mathbf{x})(i,k),
\end{align*}
which models the susceptibility of agent $\mathrm{v}_i$ to source $\mathrm{u}_k$ in layer $l$.

\item Finally,
\begin{align*}
V_l(\mathbf{x})
= I-A_l\bigl(I+\Sigma_{B_l(\mathbf{x})}\bigr),
\quad
\Sigma_{B_l(\mathbf{x})}
\triangleq \diag\bigl(B_l(\mathbf{x})\mathbf{1}_m\bigr)
\end{align*}
so that, under the normalization assumptions introduced below, \eqref{eq:ABUR} defines a convex-combination update rule.
\end{enumerate}

We next define the concatenated matrices
\begin{align*}
&A \triangleq \blkdiag(A_1,\dots,A_q), \quad  V(\mathbf{x}) = I - A\left(I+{\Sigma_{{B}}}(\mathbf{x})\right), 
 \\
&W\triangleq [W_{l,j}]_{l,j=1}^q, \quad \Sigma_{B(\mathbf{x})} \triangleq \blkdiag\!\big(\Sigma_{B_1(\mathbf{x})},\dots,\Sigma_{B_q(\mathbf{x})}\big), \\ &W_{l,j} \triangleq \Lambda_{l,j}W_j, \quad B(\mathbf{x}) \triangleq \blkdiag(B_1(\mathbf{x}),\dots,B_q(\mathbf{x})), 
\end{align*}
and write \eqref{eq:ABUR} in the stacked form
\begin{align}
\mathbf{x}[t+1]
&= f(\mathbf{x}[t])
\triangleq
V(\mathbf{x}[t])\mathbf{s}
+A\bigl(W\mathbf{x}[t]+B(\mathbf{x}[t])\mathbf{y}\bigr).
\label{eq:CBUR}
\end{align}
The block matrix $W$ is the supra-social operator of the topic-layered system.

Throughout the paper, we make the following assumptions.

\begin{asm}
For every $l\in[q]$ and $\mathbf{x}\in\mathbb{I}^{nq}$, the matrix $B_l(\mathbf{x})$ is nonnegative.
\label{asm:Positivity}
\end{asm}

\begin{asm}
For each $l\in[q]$, the matrix $W_l$ is row-stochastic.
\label{asm:Wrow}
\end{asm}

\begin{asm}
The matrix $A\bigl(I+\Sigma_{B(\mathbf{x})}\bigr)$ is row-sub-stochastic for every $\mathbf{x}\in\mathbb{I}^{nq}$.
\label{asm: ConvexCondition2}
\end{asm}

\begin{asm}
Each $b_{l,i,k}$ is Lipschitz on $\mathbb{I}^{nq}$ with respect to $\|\cdot\|_\infty$. Thus, for some $\mu_{l,i,k}\ge 0$,
\begin{align*}
|b_{l,i,k}(\mathbf{x})-b_{l,i,k}(\mathbf{z})|
\le
\mu_{l,i,k}\,\|\mathbf{x}-\mathbf{z}\|_\infty
\end{align*}
for all $\mathbf{x},\mathbf{z}\in\mathbb{I}^{nq}$.
\label{asm: LC}
\end{asm}

\begin{asm}
The contraction factor
\begin{align*}
\kappa
\triangleq
\max_{l\in[q],\,i\in[n]}
\alpha_l(i)\left(
1+
\sum_{k=1}^{m}\mu_{l,i,k}\,|\mathbf{y}_l(k)-\mathbf{s}_l(i)|
\right)
\end{align*}
satisfies $\kappa<1$.
\label{asm: CC}
\end{asm}

\begin{rem}
Under Assumptions~\ref{asm:Positivity}, \ref{asm:Wrow}, and \ref{asm: ConvexCondition2}, the map $f$ sends $\mathbb{I}^{nq}$ into itself.
\end{rem}

\begin{rem}
The term $\sum_{j=1}^{q}\Lambda_{l,j}W_j\mathbf{x}_j[t]$ represents cross-topic social evidence aggregation: when agent $\mathrm{v}_i$ updates opinion component $l$, it may use neighbors' opinions on other components $j$ as evidence. This is different from an internal self-coupling of agent $i$'s own coordinates. If $\Lambda_{l,j}=0$ for all $l\neq j$ and each mismatch weight is layer-local, e.g., $c_{l,i,k}=e_l$, then the model reduces to $q$ independent one-dimensional FJ-type systems.
\end{rem}

The significance of vector-valued modeling is illustrated in the following example.

\begin{figure}[ht]
  \centering
\begin{tikzpicture}[scale=1, every node/.style={font=\small}]

\fill[yellow!80,opacity=0.9] (-1,3.5) -- (4.2,3) -- (5.2,5) -- (0,5.5) -- cycle;
\node[rotate=-5] at (0.5,5.0) {\textbf{Layer 1}};

\node[circle,fill=blue!20,draw] (v1a) at (0,4.0) {$\mathrm{v}_1$};
\node[circle,fill=blue!20,draw] (v2a) at (4,4) {$\mathrm{v}_2$};
\node[circle,fill=red!20,draw] (u1a) at (2.5,4.6) {$\mathrm{u}_1$};

\draw[->] (v1a) edge[loop left] node[above] {0.5} ();
\draw[->] (v2a) edge[loop right] node[above] {0.5} ();
\draw[->] (v2a) -- (v1a);
\draw[->] (v1a) -- node[above] {0.5} (v2a);

\draw[->] (u1a) -- node[above] {$b_{1,2,1}$} (v2a);
\draw[->] (u1a) -- node[above] {$b_{1,1,1}$} (v1a);

\fill[green!90,opacity=0.9] (-1.0,1.0) -- (4.2,0.5) -- (5.2,2.5) -- (0,3) -- cycle;
\node[rotate=-5] at (0.5,2.5) {\textbf{Layer 2}};
\node[rotate=-5] at (2.0,3.5) {$\Lambda_{1,1}=\diag\{1,1\}$};

\node[circle,fill=blue!20,draw] (v1b) at (0,1.5) {$\mathrm{v}_1$};
\node[circle,fill=blue!20,draw] (v2b) at (4,1.5) {$\mathrm{v}_2$};
\node[circle,fill=red!20,draw] (u1b) at (2.5,2.3) {$\mathrm{u}_1$};

\draw[->] (v1b) edge[loop left] node[above] {0.5} ();
\draw[->] (v2b) edge[loop right] node[above] {0.5} ();
\draw[->] (v2b) -- (v1b);
\draw[->] (v1b) -- node[above] {0.5} (v2b);

\draw[->] (u1b) -- node[above] {$b_{2,2,1}$} (v2b);
\draw[->] (u1b) -- node[above] {$b_{2,1,1}$} (v1b);

\draw[->, thick, orange] (-1,1) -- node[midway, above, sloped, text=black] {$\Lambda_{1,2}=\mathbf{0}$} (-1,3.5);
\draw[->, thick, orange] (5.2,5) -- node[midway, above, sloped, text=black] {$\Lambda_{2,1}=\mathbf{0}$} (5.2,2.5);
\node[rotate=-5] at (2.0,1.0) {$\Lambda_{2,2}=\diag\{1,1\}$};
\end{tikzpicture}
\caption{Network used in Example~\ref{ex: Example 1} with two agents $\mathrm{v}_1$ and $\mathrm{v}_2$, and one information source $\mathrm{u}_1$.}
\label{fig:Example1_Graph}
\end{figure}
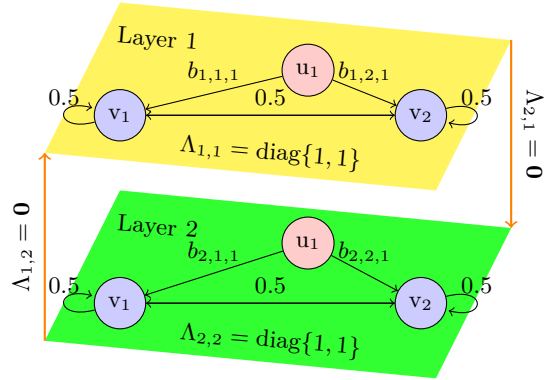

\begin{ex}
Consider $n=q=2$ and $m=1$. Let
\begin{align*}
W_1=W_2
&=
\begin{bmatrix}
0.5 & 0.5\\
0.5 & 0.5
\end{bmatrix},
&A_1=A_2&=\diag(0.8,0.8),\\
\Lambda_{1,1}=\Lambda_{2,2}
&=
\begin{bmatrix}
1 & 0\\
0 & 1
\end{bmatrix},
&\Lambda_{1,2}=\Lambda_{2,1}&=\mathbf{0}.
\end{align*}
Let
\begin{align*}
\hat{\mathbf{s}}_1&=\begin{bmatrix}0 & 0\end{bmatrix}^{\top},
&\hat{\mathbf{s}}_2&=\begin{bmatrix}0 & 0.5\end{bmatrix}^{\top},
&\hat{\mathbf{y}}_1&=\begin{bmatrix}1 & 1\end{bmatrix}^{\top}.
\end{align*}
Choose
\begin{align*}
b_{1,i,1}(\mathbf{x})
&=0.24-0.24\begin{bmatrix}0.95 & 0.05\end{bmatrix}|\hat{\mathbf{x}}_i-\hat{\mathbf{y}}_1|,\\
b_{2,i,1}(\mathbf{x})
&=0.24-0.24\begin{bmatrix}0.05 & 0.95\end{bmatrix}|\hat{\mathbf{x}}_i-\hat{\mathbf{y}}_1|.
\end{align*}
Now consider the decoupled approximation
\begin{align*}
\hat{b}_{1,i,1}(\mathbf{x})
&=0.24-0.24\begin{bmatrix}1 & 0\end{bmatrix}|\hat{\mathbf{x}}_i-\hat{\mathbf{y}}_1|,\\
\hat{b}_{2,i,1}(\mathbf{x})
&=0.24-0.24\begin{bmatrix}0 & 1\end{bmatrix}|\hat{\mathbf{x}}_i-\hat{\mathbf{y}}_1|.
\end{align*}
Since $\Lambda_{1,2}=\Lambda_{2,1}=0$, this example isolates \emph{source-side} multidimensionality: the two systems differ only through the cross-coordinate weights in the source terms. The second system consists of two independent scalar dynamics, one for each layer. The first system cannot be represented in that way. Figure~\ref{fig:Example1_Simulation} shows that the two systems align well only in the second layer, which confirms that separate one-dimensional models do not capture the full multidimensional behavior.
\label{ex: Example 1}
\end{ex}

\begin{figure}[t]
\centering
\includegraphics[scale=0.24]{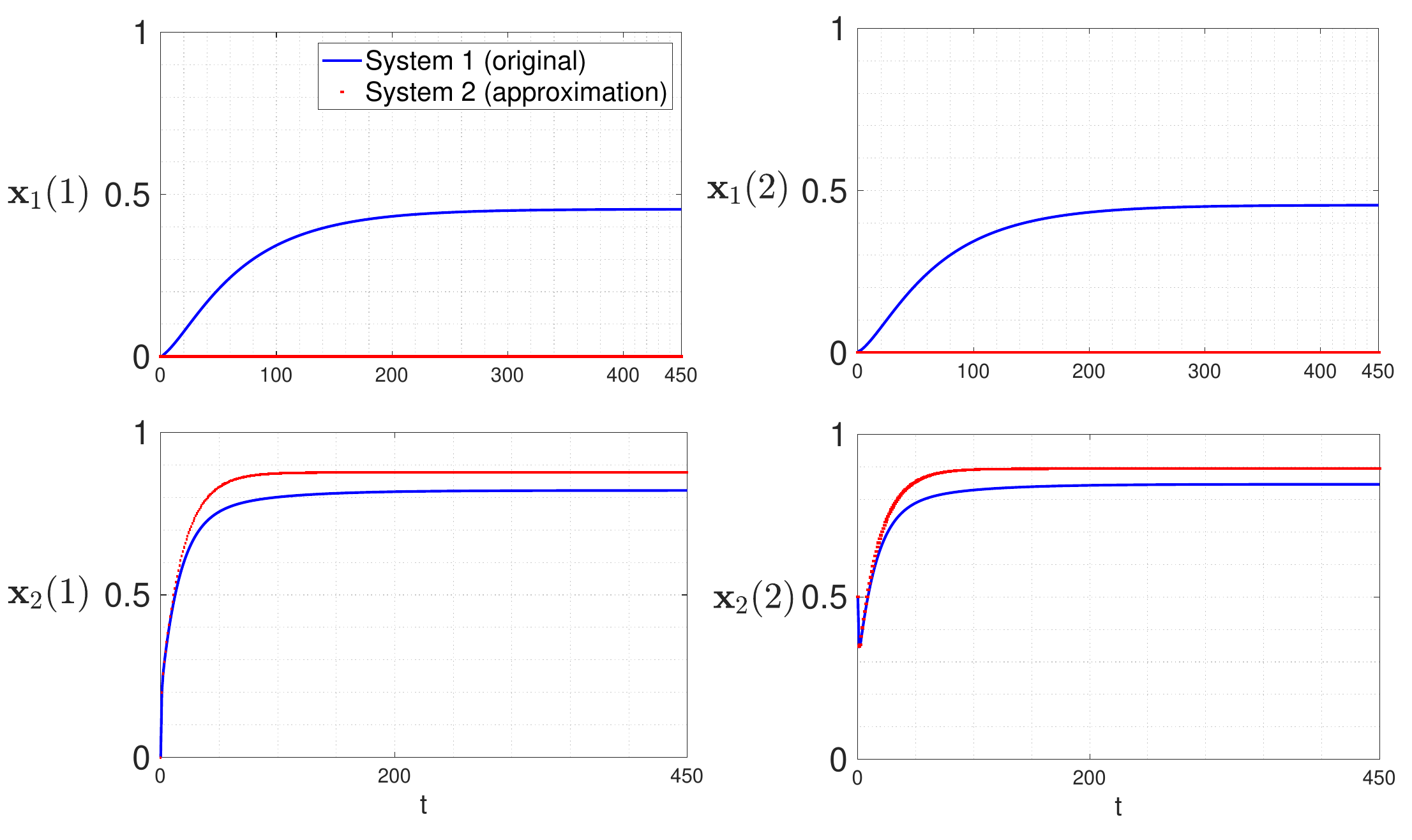}
\caption{Example~\ref{ex: Example 1}: true two-dimensional dynamics versus decoupled scalar approximation.}
\label{fig:Example1_Simulation}
\end{figure}

\begin{rem}
Equation~\eqref{eq:CBUR} can be written as
\begin{align*}
\mathbf{x}[t+1]
&=
\bigl(I-A(I+\Sigma_{B(\mathbf{x}[t])})\bigr)\mathbf{s}
+A\bigl(W\mathbf{x}[t]+B(\mathbf{x}[t])\mathbf{y}\bigr).
\end{align*}
The diagonal term $\Sigma_{B(\mathbf{x})}$ shifts weight from innate opinions to source influence. Larger source weights therefore reduce the effective resistance to external information.
\label{thm: rmark4}
\end{rem}

We next define a source--agent mismatch score and use it to model $B_l(\mathbf{x})$ in two relevant settings. For each $l\in[q]$, define the matrix $\beta_l\in\mathbb{R}_+^{n\times m}$ entrywise by
\begin{equation}
\beta_l(i,k)
\triangleq
c_{l,i,k}^{\top}|H_i\mathbf{x}-T_k\mathbf{y}|,
\label{eq:beta}
\end{equation}
where $c_{l,i,k}\in\mathbb{I}^q$ is nonnegative and satisfies $\mathbf{1}^{\top}c_{l,i,k}=1$. A small value of $\beta_l(i,k)$ means that agent $i$ is close to source $k$ and is therefore more receptive.

The first setting corresponds to a bounded-confidence function, where
\begin{equation}
b_{l,i,k}(\mathbf{x}) =
\begin{cases}
P_l(i,k), & \beta_l(i,k)\le \varepsilon_i,\\
0, & \text{otherwise},
\end{cases}
\label{eq: bounded confidence}
\end{equation}
while the second corresponds to an affine confirmation-bias function,
\begin{equation}
B_l(\mathbf{x}) = \Omega_l - \Gamma_l\circ \beta_l,
\label{eq: Affine function}
\end{equation}
where $P_l,\Omega_l,\Gamma_l\in\mathbb{R}_+^{n\times m}$ and $\varepsilon_i>0$.

For the affine rule, it is convenient to state explicit parameter-level sufficient conditions. Because $H_i\mathbf{x},T_k\mathbf{y}\in\mathbb{I}^q$ and $c_{l,i,k}$ is nonnegative with $\mathbf{1}^{\top}c_{l,i,k}=1$, we have
\begin{align*}
0\le \beta_l(i,k)\le 1
\end{align*}
for all $\mathbf{x}\in\mathbb{I}^{nq}$. Hence, if for every $\alpha_l(i)\neq 0$,
\begin{align}
0 \le \Gamma_l(i,k) \le \Omega_l(i,k),
\qquad
\sum_{k=1}^{m}\Omega_l(i,k)
\le
\frac{1-\alpha_l(i)}{\alpha_l(i)},
\label{eq:AffineSufficient1}
\end{align}
then
\begin{align*}
0\le b_{l,i,k}(\mathbf{x})\le \Omega_l(i,k).
\end{align*}
Therefore $B_l(\mathbf{x})$ is nonnegative and $A\bigl(I+\Sigma_{B(\mathbf{x})}\bigr)$ is row-sub-stochastic on $\mathbb{I}^{nq}$. Moreover, for the affine rule,
\begin{align*}
|b_{l,i,k}(\mathbf{x})-b_{l,i,k}(\mathbf{z})|
\le
\Gamma_l(i,k)\,\|\mathbf{x}-\mathbf{z}\|_\infty.
\end{align*}
A convenient stronger sufficient condition for Assumption~\ref{asm: CC} is
\begin{align}
\max_{l\in[q],\,i\in[n]}
\alpha_l(i)\left(
1+
\sum_{k=1}^{m}\Gamma_l(i,k)
\right)
<1.
\label{eq:AffineSufficient2}
\end{align}
We will enforce these uniform conditions in Section~\ref{Real Network Simulation}.

The next numerical example illustrates different behaviors under the bounded-confidence rule \eqref{eq: bounded confidence} and the affine rule \eqref{eq: Affine function} in a network of three agents discussing two interconnected issues.

\begin{figure}[t]
\centering
\begin{tikzpicture}[scale=1, every node/.style={font=\small}]

\fill[yellow!90,opacity=0.9] (-1,3.5) -- (4.2,3) -- (5.2,5) -- (0,5.5) -- cycle;
\node[rotate=-5] at (0.5,5.2) {\textbf{Layer 1}};
\node[rotate=-4] at (1.5,3.4) {$\Lambda_{1,1}=\diag\{0.7,0.5,0.5\}$};

\node[circle,fill=blue!20,draw] (v1a) at (0,4.0) {$\mathrm{v}_1$};
\node[circle,fill=blue!20,draw] (v2a) at (4.2,4.5) {$\mathrm{v}_2$};
\node[circle,fill=blue!20,draw] (v3a) at (3.8,3.5) {$\mathrm{v}_3$};
\node[circle,fill=red!20,draw] (u1a) at (2.6,4.8) {$\mathrm{u}_1$};

\draw[->] (v1a) edge[loop left] node[above] {0.2} ();
\draw[->] (v2a) edge[loop right] node[above] {1.0} ();
\draw[->] (v2a) -- node[above] {0.4} (v1a);
\draw[->] (v3a) -- node[above] {0.4} (v1a);
\draw[->] (v2a) -- node[right] {1.0} (v3a);

\draw[->] (u1a) -- node[above] {$b_{1,2,1}$} (v2a);

\fill[green!90,opacity=0.9] (-1.0,0.5) -- (4.2,0.0) -- (5.2,2.0) -- (0,2.5) -- cycle;
\node[rotate=-5] at (0.5,2.2) {\textbf{Layer 2}};
\node[rotate=-4] at (1.5,0.4) {$\Lambda_{2,2}=\diag\{0.7,0.5,0.5\}$};

\node[circle,fill=blue!20,draw] (v1b) at (0,1.0) {$\mathrm{v}_1$};
\node[circle,fill=blue!20,draw] (v2b) at (4.2,1.5) {$\mathrm{v}_2$};
\node[circle,fill=blue!20,draw] (v3b) at (3.8,0.5) {$\mathrm{v}_3$};
\node[circle,fill=red!20,draw] (u1b) at (2.6,1.8) {$\mathrm{u}_1$};

\draw[->] (v2b) edge[loop right] node[above] {1.0} ();
\draw[->] (v3b) edge[loop right] node[above] {1.0} ();
\draw[->] (v2b) -- node[above] {1.0} (v1b);

\draw[->] (u1b) -- node[above] {$b_{2,2,1}$} (v2b);

\draw[->, thick, orange] (-1,0.5) -- node[midway, above, sloped, text=black] {$\Lambda_{1,2}=\diag\{0.3,0.5,0.5\}$} (-1,3.5);
\draw[->, thick, orange] (5.2,5) -- node[midway, above, sloped, text=black] {$\Lambda_{2,1}=\diag\{0.3,0.5,0.5\}$} (5.2,2.0);
\end{tikzpicture}
\caption{Network used in Example~\ref{ex: Example 2} with three agents $\mathrm{v}_1$, $\mathrm{v}_2$, and $\mathrm{v}_3$, and one information source $\mathrm{u}_1$.}
\label{fig:Example2_graph}
\end{figure}
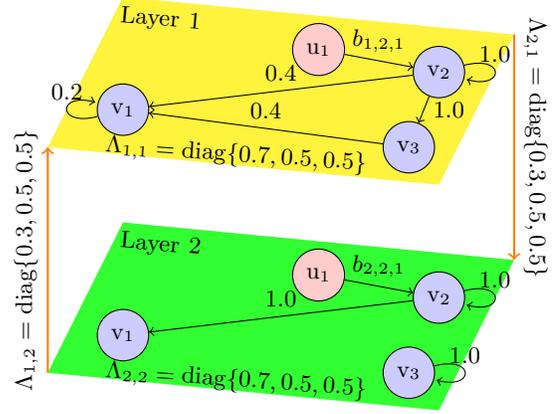

\begin{ex}
Consider three agents, one source, and two layers as shown in Fig.~\ref{fig:Example2_graph}. Let
\begin{align*}
W_1
&=
\begin{bmatrix}
0.2 & 0.4 & 0.4\\
0 & 1 & 0\\
0 & 1 & 0
\end{bmatrix},
&W_2
&=
\begin{bmatrix}
0 & 1 & 0\\
0 & 1 & 0\\
0 & 0 & 1
\end{bmatrix},
\end{align*}
\begin{align*}
A_1&=\diag(0.5,0.8,0.5),
&A_2&=\diag(0.5,0.6,0.5),
\end{align*}
\begin{align*}
\Lambda_{1,1}&=\Lambda_{2,2}=\diag(0.7,0.5,0.5), \\
\Lambda_{1,2}&=\Lambda_{2,1}=\diag(0.3,0.5,0.5).
\end{align*}
Let
\begin{align*}
\hat{\mathbf{s}}_1\!=\!\begin{bmatrix}0,\, 1\end{bmatrix}^{\top}\!, \,\,
\hat{\mathbf{s}}_2\!=\!\begin{bmatrix}0.5,\, 0.5\end{bmatrix}^{\top}\!, \,\,
\hat{\mathbf{s}}_3\!=\!\begin{bmatrix}0,\, 0\end{bmatrix}^{\top}\!, \,\,
\hat{\mathbf{y}}_1\!=\!\begin{bmatrix}1, \, 1\end{bmatrix}^{\top}\!.
\end{align*}
Set
\begin{align*}
\beta_1(1,1)&=\beta_1(3,1)=0,\qquad 
\beta_2(1,1)=\beta_2(3,1)=0,\\
\beta_1(2,1)
&=\begin{bmatrix}0.8  0.2\end{bmatrix}|\hat{\mathbf{x}}_2-\hat{\mathbf{y}}_1|, \\ 
\beta_2(2,1)
&=\begin{bmatrix}0.2 & 0.8\end{bmatrix}|\hat{\mathbf{x}}_2-\hat{\mathbf{y}}_1|.
\end{align*}
For the bounded-confidence rule, choose
\begin{align*}
P_1=P_2=\begin{bmatrix}0 & 0.25 & 0\end{bmatrix}^{\top}.
\end{align*}
Figure~\ref{fig:Example2_1} depicts $\mathbf{x}_1(2)$ and $\mathbf{x}_2(2)$ for $\varepsilon_2\in\{0,0.2,0.3,0.5\}$. When $\varepsilon_2=0$, both coordinates converge to a fixed point. Periodic orbits appear for intermediate thresholds. For $\varepsilon_2=0.5$, the system again converges to a fixed point.

For the affine rule, choose
\begin{align*}
\Gamma_1=\Gamma_2
&=
\gamma\begin{bmatrix}0 & 0.5 & 0\end{bmatrix}^{\top},
&\Omega_1=\Omega_2
&=
\begin{bmatrix}0 & 0.25 & 0\end{bmatrix}^{\top}.
\end{align*}
Figure~\ref{fig:Example2_2} shows a smooth transition between the two fixed-point regimes and no periodic orbits. This is consistent with the contractive affine regime studied in the next section.
\label{ex: Example 2}
\end{ex}

\begin{figure}[t]
\centering
\includegraphics[scale=0.44]{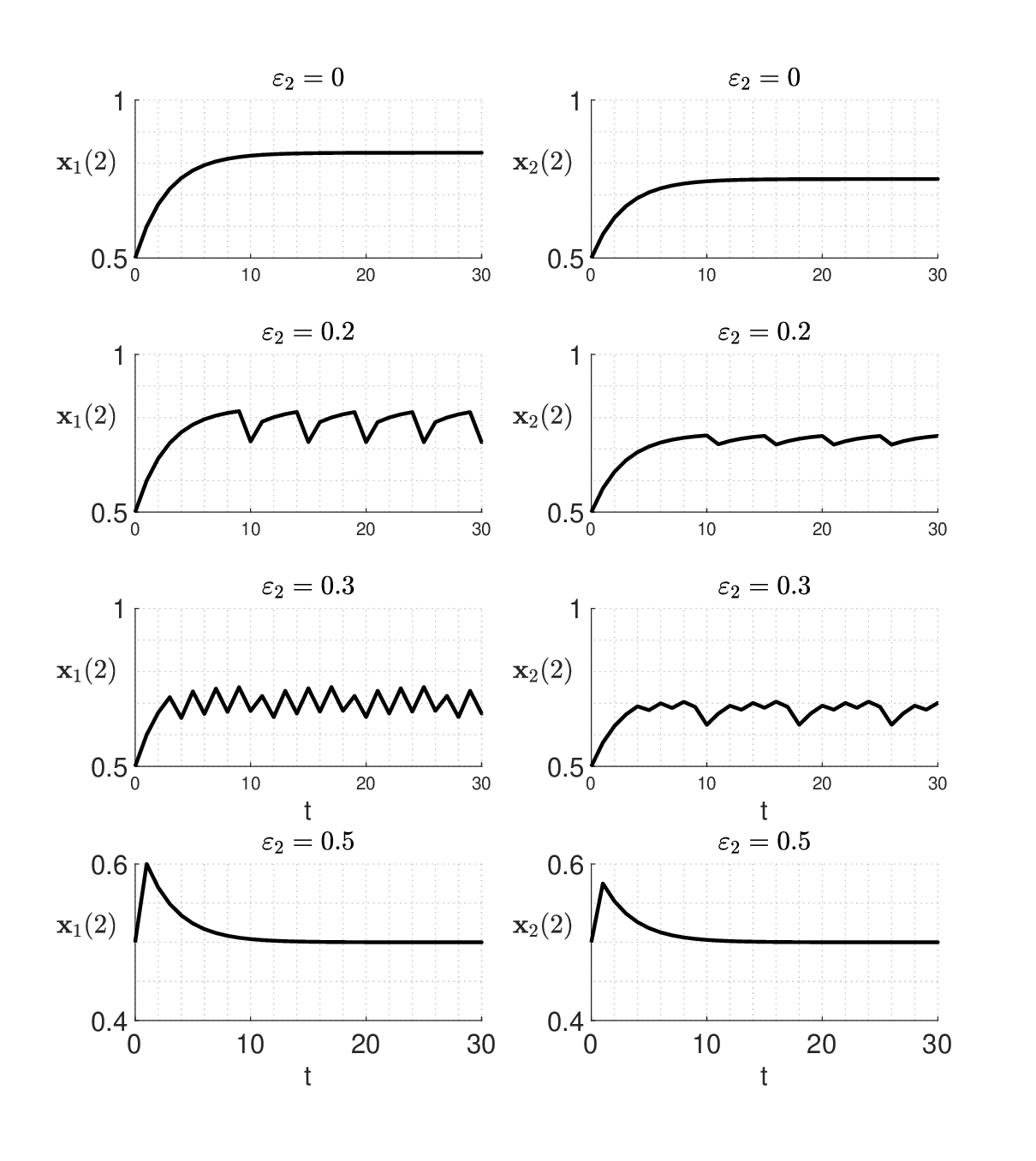}
\caption{Example~\ref{ex: Example 2}: bounded-confidence dynamics for several thresholds.}
\label{fig:Example2_1}
\end{figure}

\begin{figure}[t]
\centering
\includegraphics[scale=0.44]{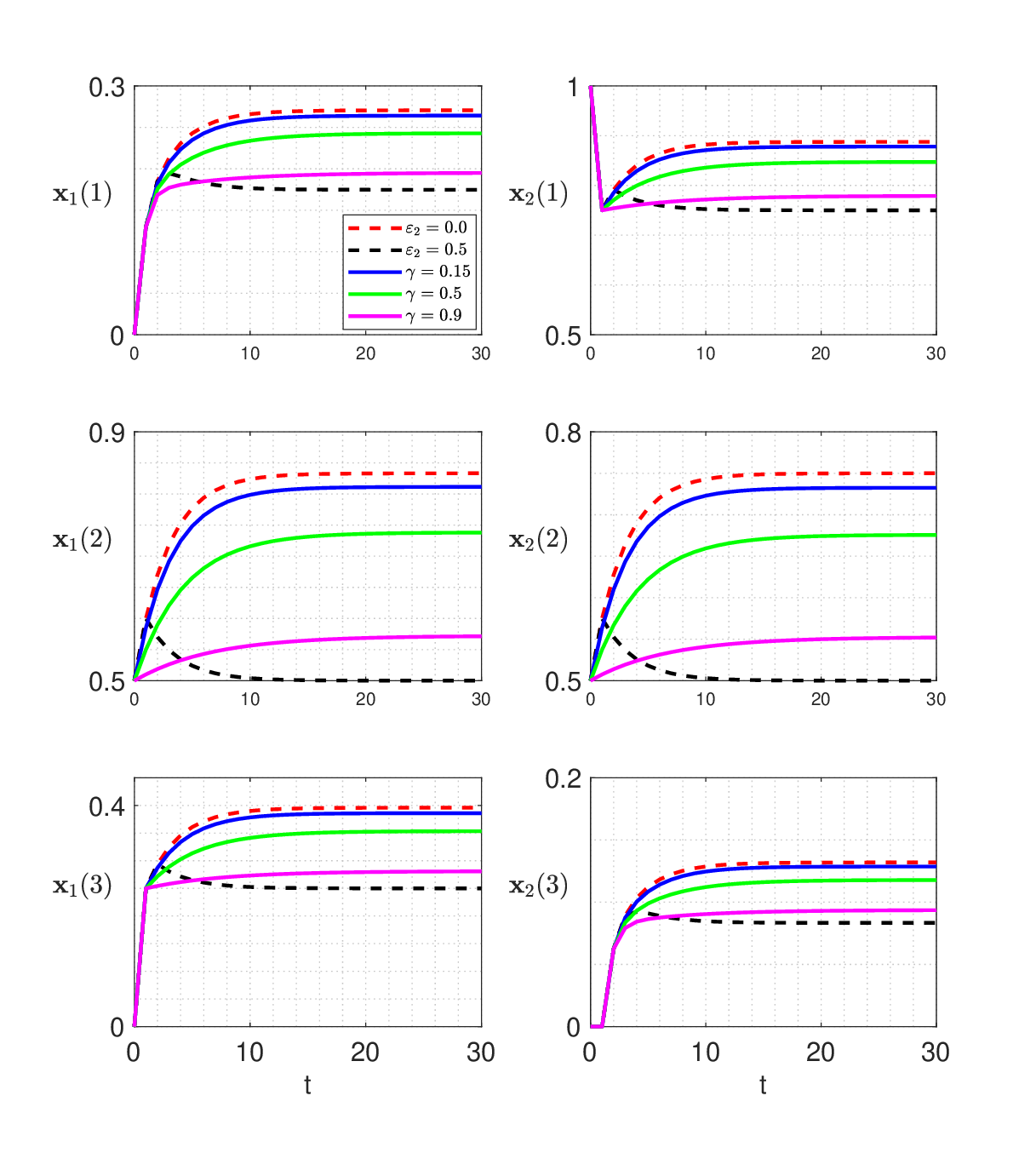}
\caption{Example~\ref{ex: Example 2}: affine confirmation-bias dynamics for several values of $\gamma$ and $\varepsilon$ combinations.}
\label{fig:Example2_2}
\end{figure}

\section{Convergence and Fixed-Point Analysis}

\subsection{Convergence Analysis}
\label{Convergence Analysis}
We next state sufficient conditions under which \eqref{eq:CBUR} is a contraction on
$\mathbb{I}^{nq}$ and hence admits a unique fixed point.

\begin{thm}
Under Assumptions~\ref{asm:Positivity}, \ref{asm:Wrow}, \ref{asm: ConvexCondition2},
\ref{asm: LC}, and \ref{asm: CC}, the map $f$ in \eqref{eq:CBUR} is a contraction on
$\mathbb{I}^{nq}$ under $\|\cdot\|_\infty$. Hence \eqref{eq:CBUR} admits a unique
fixed point $\mathbf{x}^*\in\mathbb{I}^{nq}$, and every trajectory converges to it.
The fixed point satisfies
\begin{equation}
\mathbf{x}^* = (I-A)\mathbf{s} + A\left(W\mathbf{x}^* + B(\mathbf{x}^*)\mathbf{y}
- \Sigma_{B(\mathbf{x}^*)}\mathbf{s}\right).
\label{eq:steady-state}
\end{equation}
\label{thm:th1}
\end{thm}

\begin{rem}
The choice of $\|\cdot\|_\infty$ is natural here. The social term is row-wise
averaging, so $\|\cdot\|_\infty$ is controlled by row sums, whereas $\ell_1$ is
controlled by column sums. The nonlinear source term also admits coordinatewise
Lipschitz bounds, which $\|\cdot\|_\infty$ turns into a maximum rather than a sum.
Other norms could be used in principle, but then the Lipschitz and contraction
conditions would have to be restated in that norm.
\end{rem}

\begin{rem}
Theorem~\ref{thm:th1} describes a globally contractive regime. In this regime, the
model cannot exhibit multiple attracting fixed points, periodic orbits, or HK-style
fragmentation driven by discontinuous confidence sets. Such behaviors can arise only
when the contraction condition is violated or when discontinuous mechanisms such as
\eqref{eq: bounded confidence} are used instead. In particular, the bounded-confidence
rule \eqref{eq: bounded confidence} is not covered by Theorem~\ref{thm:th1}.
\end{rem}

The following example illustrates convergence to a unique fixed point from arbitrary
initial conditions.

\begin{ex}
Take the parameters of Example~\ref{ex: Example 2} with the affine rule and
$\gamma=0.5$. They satisfy the row-wise contraction condition of
Theorem~\ref{thm:th1}; indeed, the largest row-wise factor is
\begin{align*}
0.8+0.8\times 0.25\times 0.5 = 0.9 < 1.
\end{align*}
Figure~\ref{fig:Example3} shows $500$ runs from random initial conditions. All
trajectories converge to the same fixed point.
\label{ex: Example 3}
\end{ex}

\begin{figure}[t]
\centering
\includegraphics[scale=0.5]{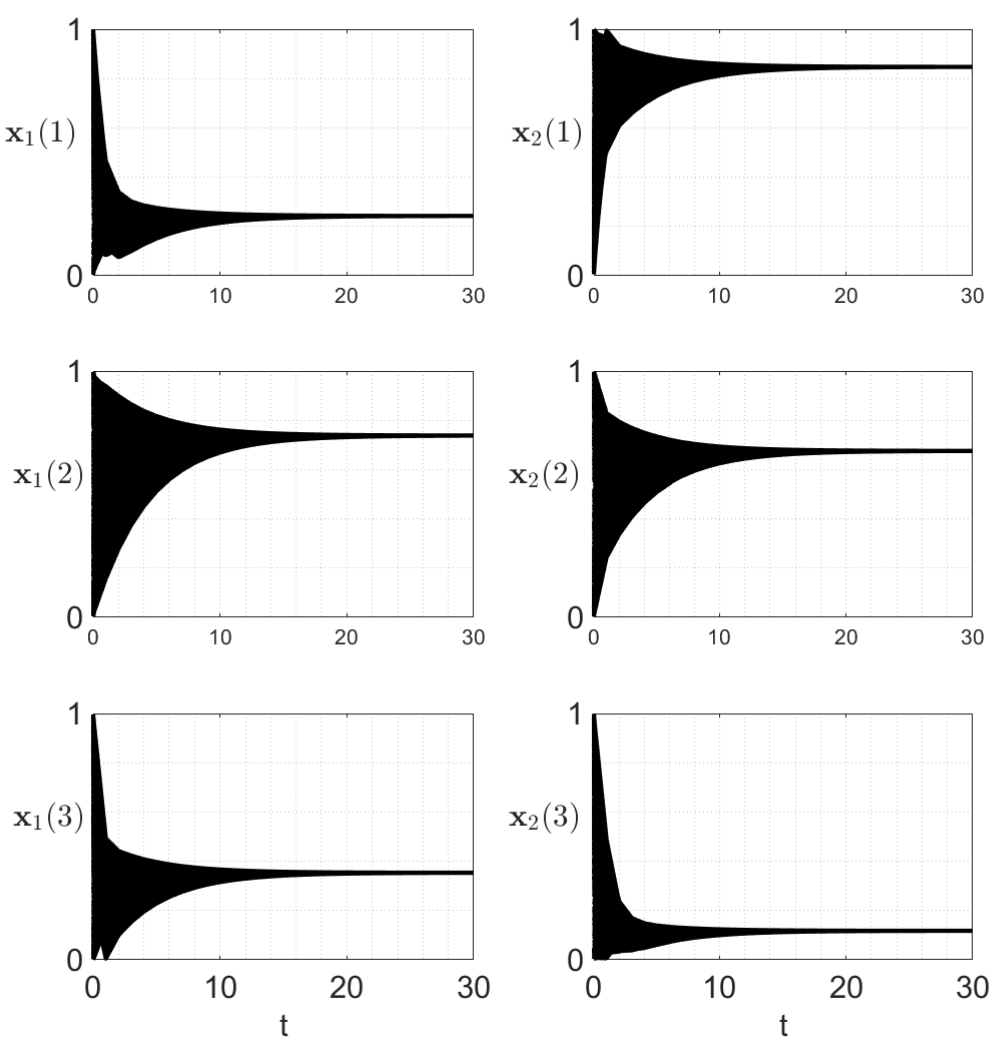}
\caption{Example~\ref{ex: Example 3}: convergence to one fixed point from $500$ random initial conditions.}
\label{fig:Example3}
\end{figure}

\subsection{Fixed-Point Analysis}
\label{section: Fixed Point Analysis}
We now analyze the fixed point. We first isolate the nonlinear fixed-point equation,
and then treat the affine and bounded cases.

\begin{cor}
Let $Z\triangleq (I-AW)^{-1}$. Under Theorem~\ref{thm:th1}, $Z$ is well defined and
the fixed point satisfies
\begin{equation}
\mathbf{x}^* = Z\Bigl((I-A)\mathbf{s} + AB(\mathbf{x}^*)\mathbf{y}
- A\Sigma_{B(\mathbf{x}^*)}\mathbf{s}\Bigr).
\label{eq:EPPNA}
\end{equation}
\label{thm:cor2}
\end{cor}

\begin{proof}
Because each $W_j$ is row-stochastic and $\sum_{j=1}^q\Lambda_{l,j}=I$, the
$((l-1)n+i)$th row sum of $AW$ equals $\alpha_l(i)$. Hence
\begin{align*}
\|AW\|_\infty = \max_{l\in[q],\,i\in[n]} \alpha_l(i) < 1,
\end{align*}
where the strict inequality follows from Assumption~\ref{asm: CC}. Therefore
$\rho(AW)\le \|AW\|_\infty<1$, so $I-AW$ is nonsingular. Rearranging
\eqref{eq:steady-state} gives \eqref{eq:EPPNA}.
\end{proof}

\medskip
\subsubsection{Fixed Point Under the Affine Rule}
\label{subsection: Steady-State Condition in Presence of Affine Functions}
For the affine rule, the only nonlinearity in \eqref{eq:CBUR} comes from the sign
pattern of $H_i\mathbf{x}-T_k\mathbf{y}$. Once that pattern is fixed, the dynamics
become affine.

\begin{thm}
Assume \eqref{eq: Affine function}, Assumption~\ref{asm:Wrow}, and the parameter
condition \eqref{eq:AffineSufficient1}. Suppose also that
\begin{align}
\kappa_{\mathrm{aff}}
\triangleq
\max_{l\in[q],\,i\in[n]}
\alpha_l(i)\left(
1+\sum_{k=1}^m \Gamma_l(i,k)\,
|\mathbf{y}_l(k)-\mathbf{s}_l(i)|
\right)<1.
\label{eq:kappa-aff}
\end{align}
Then the fixed point $\mathbf{x}^*$ of \eqref{eq:CBUR} can be found by a finite
consistency search over sign patterns $\boldsymbol{\theta}_{i,k}\in\{-1,0,1\}^q$.
For each candidate pattern $\boldsymbol{\theta}$, the induced affine system has the
form
\begin{align*}
\mathbf{x}[t+1]
=
A_a\mathbf{s}
+
B_a(\boldsymbol{\theta})\mathbf{y}
+
W_a(\boldsymbol{\theta})\mathbf{x}[t],
\end{align*}
and satisfies
\begin{align*}
\|W_a(\boldsymbol{\theta})\|_\infty \le \kappa_{\mathrm{aff}}<1.
\end{align*}
Hence $I-W_a(\boldsymbol{\theta})$ is nonsingular for every candidate
$\boldsymbol{\theta}$. Consequently,
Algorithm~\ref{algo: Fixed Point Affine Functions} is well defined and returns the
unique fixed point after finitely many consistency checks.
\label{thm:th2}
\end{thm}

\restylefloat{algorithm}
\begin{algorithm}[t]
\caption{Fixed-point computation under the affine rule,  $\Sigma_\Omega\triangleq \blkdiag(\Sigma_{\Omega_1},\dots,\Sigma_{\Omega_q})$ and $A_a\triangleq I-A-A\Sigma_\Omega$.}
\label{algo: Fixed Point Affine Functions}
\footnotesize
\begin{algorithmic}[1]
\renewcommand{\algorithmicrequire}{\textbf{Input:}}
\renewcommand{\algorithmicensure}{\textbf{Output:}}
\REQUIRE Finite candidate set $\mathcal{P}$ of sign patterns $\boldsymbol{\theta}=\{\boldsymbol{\theta}_{i,k}\}$.
\ENSURE The fixed point $\mathbf{x}^*$.
\FOR{each $\boldsymbol{\theta}\in\mathcal{P}$}
    \STATE Build $R^x(\boldsymbol{\theta})$ and $R^y(\boldsymbol{\theta})$ as in Appendix~\ref{Proof of Theorem 2}.
    \STATE Set $W_a(\boldsymbol{\theta})=A\bigl(W+R^x(\boldsymbol{\theta})\bigr)$.
    \STATE Set $B_a(\boldsymbol{\theta})=A\bigl(\Omega+R^y(\boldsymbol{\theta})\bigr)$.
    \STATE Compute $\mathbf{x}_{\mathrm{cand}}=(I-W_a(\boldsymbol{\theta}))^{-1}(A_a\mathbf{s}+B_a(\boldsymbol{\theta})\mathbf{y})$.
    \IF{$\sign(H_i\mathbf{x}_{\mathrm{cand}}-T_k\mathbf{y})=\boldsymbol{\theta}_{i,k}$ for all $i\le n$, $k\le m$}
        \STATE return $\mathbf{x}_{\mathrm{cand}}$.
    \ENDIF
\ENDFOR
\end{algorithmic}
\end{algorithm}

The crude worst-case bound $|\mathcal{P}|\le 3^{nmq}$ is often far from tight.
First, agent--source pairs with $\Omega_l(i,k)=\Gamma_l(i,k)=0$ contribute no source
term and can be removed from the sign search altogether. Second,
Theorem~\ref{thm:BEPNWF} provides coordinatewise intervals
$[\hat{\mathbf{x}}_{i,L}^*(l),\hat{\mathbf{x}}_{i,U}^*(l)]$ that can fix many signs in
advance; Corollary~\ref{thm:analytical} is the extreme case in which every sign is
determined and the search collapses to a single pattern. Hence the effective search
dimension is the number of \emph{active ambiguous coordinates}, not $nmq$. In the
real-network scenario studied in Section~\ref{Real Network Simulation}, one source is
connected to only four agents, so before any further pruning this leaves at most
$4\times 4=16$ relevant sign variables.

\subsubsection{Bounded Source-Influence Functions}
\label{subsection: Steady-State Condition Approximation in Presence of Bounded Functions}
Assume there exist nonnegative matrices $\underline{\Phi}_l$ and
$\overline{\Phi}_l$ such that
\begin{equation}
\underline{\Phi}_l \le B_l(\mathbf{x}) \le \overline{\Phi}_l,
\qquad \forall\,\mathbf{x}\in\mathbb{I}^{nq}.
\label{eq:bounded-state-dependent-function}
\end{equation}
Define
\begin{align*}
\underline{\Phi} &\triangleq \blkdiag(\underline{\Phi}_1,\dots,\underline{\Phi}_q),
&
\overline{\Phi} &\triangleq \blkdiag(\overline{\Phi}_1,\dots,\overline{\Phi}_q),\\
\Sigma_{\underline{\Phi}} &\triangleq \blkdiag(\Sigma_{\underline{\Phi}_1},\dots,\Sigma_{\underline{\Phi}_q}),
&
\Sigma_{\overline{\Phi}} &\triangleq \blkdiag(\Sigma_{\overline{\Phi}_1},\dots,\Sigma_{\overline{\Phi}_q}).
\end{align*}

\begin{thm}
Under Theorem~\ref{thm:th1}, the fixed point satisfies\footnote{All inequalities,
maxima, and minima below are componentwise.}
\begin{equation}
\mathbf{x}_L^* \le \mathbf{x}^* \le \mathbf{x}_U^*,
\label{eq:BEPNWF}
\end{equation}
where
\begin{align*}
\mathbf{x}_L^*
&\triangleq
\max\!\left\{Z\bigl((I-A)\mathbf{s}+A\underline{\Phi}\,\mathbf{y}
-A\Sigma_{\overline{\Phi}}\mathbf{s}\bigr),\,\mathbf{0}\right\},\\
\mathbf{x}_U^*
&\triangleq
\min\!\left\{Z\bigl((I-A)\mathbf{s}+A\overline{\Phi}\,\mathbf{y}
-A\Sigma_{\underline{\Phi}}\mathbf{s}\bigr),\,\mathbf{1}\right\}.
\end{align*}
\label{thm:BEPNWF}
\end{thm}

\begin{proof}
From \eqref{eq:EPPNA},
\begin{align*}
\mathbf{x}^* = Z\Bigl((I-A)\mathbf{s}+AB(\mathbf{x}^*)\mathbf{y}
-A\Sigma_{B(\mathbf{x}^*)}\mathbf{s}\Bigr).
\end{align*}
By \eqref{eq:bounded-state-dependent-function},
\begin{align*}
\underline{\Phi} \le B(\mathbf{x}^*) \le \overline{\Phi},
\qquad
\Sigma_{\underline{\Phi}} \le \Sigma_{B(\mathbf{x}^*)} \le \Sigma_{\overline{\Phi}}.
\end{align*}
Also,
\begin{align*}
Z=(I-AW)^{-1}=\sum_{r=0}^{\infty}(AW)^r
\end{align*}
is nonnegative because $AW\ge 0$ and $\rho(AW)<1$. Therefore multiplication by $Z$
preserves elementwise inequalities. Clipping to $[\mathbf{0},\mathbf{1}]$ preserves the
bounds because $\mathbf{x}^*\in\mathbb{I}^{nq}$.
\end{proof}

\begin{ex}
In Example~\ref{ex: Example 2}, suppose we replace $b_{l,2,1}(\mathbf{x})$ by
\begin{align*}
{b}_{1,2,1}(\mathbf{x}) &= 0.55\ln\!\bigl(2-\beta_1(2,1)\bigr), \\
{b}_{2,2,1}(\mathbf{x}) &= 0.25 - 0.25 \sin\!\bigl(\beta_2(2,1)\bigr).
\end{align*}
Since $\beta_1(2,1),\beta_2(2,1)\in[0,1]$, one admissible choice of bounding matrices
is
\begin{align*}
\underline{\Phi}_1 &= \mathbf{0}_3,
&
\underline{\Phi}_2 &= \begin{bmatrix} 0 & 0.0396 & 0 \end{bmatrix}^{\top}, \\
\overline{\Phi}_1 &= \begin{bmatrix} 0 & 0.3812 & 0 \end{bmatrix}^{\top},
&
\overline{\Phi}_2 &= \begin{bmatrix} 0 & 0.2500 & 0 \end{bmatrix}^{\top}.
\end{align*}
Figure~\ref{fig:Example4} compares the fixed point with the bounds from
Theorem~\ref{thm:BEPNWF}.
\label{ex:Nonlinear_Weight_Function_Example}
\end{ex}

\begin{figure}[t]
\centering
\includegraphics[scale=0.28]{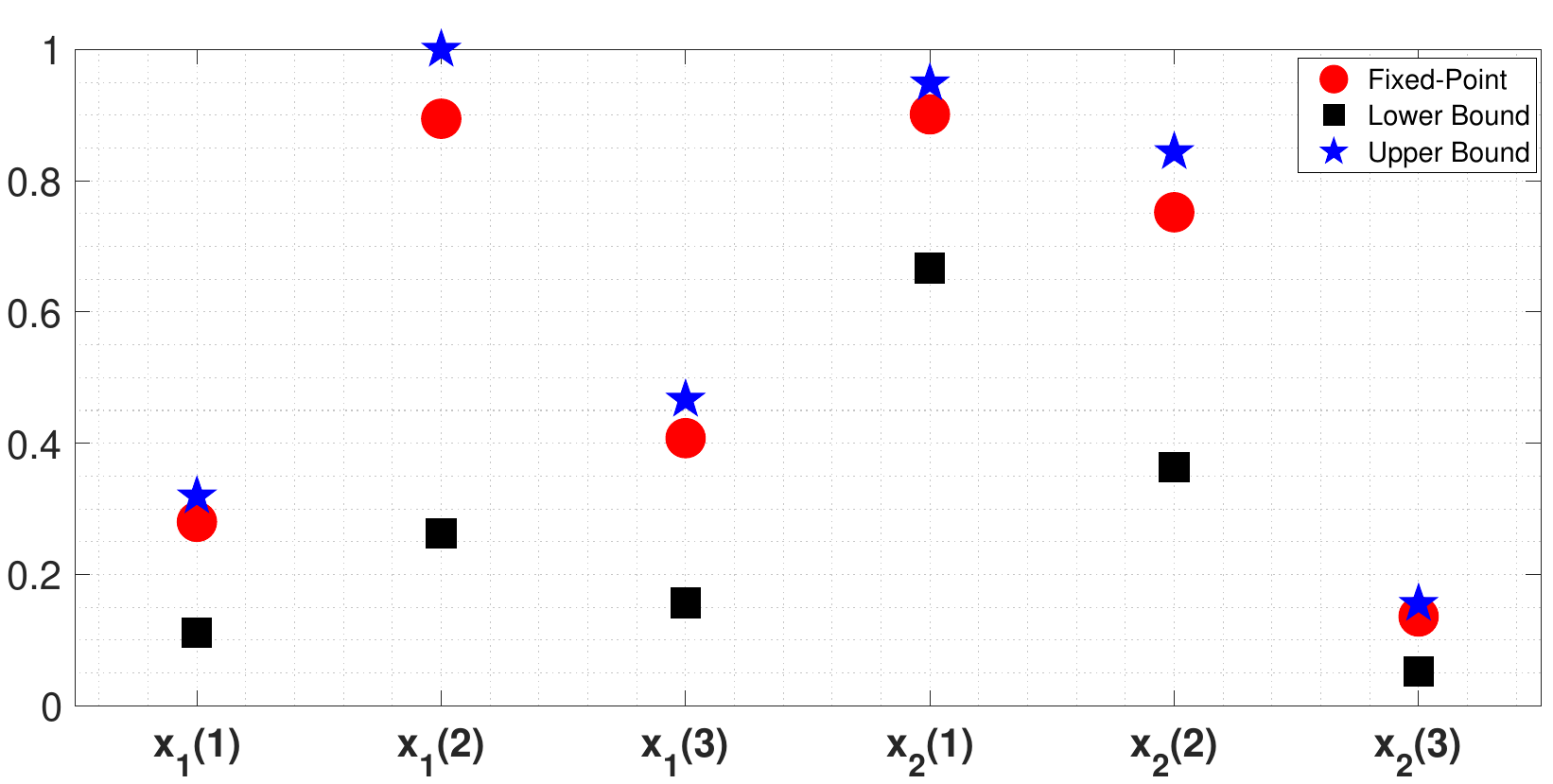}
\caption{Example~\ref{ex:Nonlinear_Weight_Function_Example}:fixed point and componentwise bounds.}
\label{fig:Example4}
\end{figure}

\subsubsection{Special Case with a Closed-Form Solution}
\label{Closed-form Solutions in Presence of Affine Weight Functions}
In general, the affine rule requires a finite sign search. In a special regime,
however, the sign pattern is fixed a priori and the fixed point admits a closed-form
expression.

\begin{cor}
Assume \eqref{eq: Affine function}. Let
$\hat{\mathbf{x}}^*_{i,U}\triangleq H_i\mathbf{x}_U^*$ and
$\hat{\mathbf{x}}^*_{i,L}\triangleq H_i\mathbf{x}_L^*$. If, for every $i\le n$,
$k\le m$, and $l\le q$,
\begin{equation}
\hat{\mathbf{y}}_k(l) > \hat{\mathbf{x}}^*_{i,U}(l)
\quad \text{or} \quad
\hat{\mathbf{y}}_k(l) < \hat{\mathbf{x}}^*_{i,L}(l),
\label{eq:cor2}
\end{equation}
then the sign pattern is unique. Let $\boldsymbol{\theta}$ denote that unique pattern,
and define
\begin{align*}
A_a &\triangleq I-A-A\Sigma_{\Omega},\\
W_a &\triangleq A\bigl(W+R^x(\boldsymbol{\theta})\bigr),\\
B_a &\triangleq A\bigl(\Omega+R^y(\boldsymbol{\theta})\bigr),
\end{align*}
where $R^x(\boldsymbol{\theta})$ and $R^y(\boldsymbol{\theta})$ are defined as in
Appendix~\ref{Proof of Theorem 2}. Then the fixed point has the closed form
\begin{align*}
\mathbf{x}^* = (I-W_a)^{-1}(A_a\mathbf{s}+B_a\mathbf{y}).
\end{align*}
\label{thm:analytical}
\end{cor}

\begin{proof}
Theorem~\ref{thm:BEPNWF} yields
\begin{align*}
\hat{\mathbf{x}}^*_{i,L}(l) \le \hat{\mathbf{x}}_i^*(l) \le \hat{\mathbf{x}}^*_{i,U}(l).
\end{align*}
Condition~\eqref{eq:cor2} fixes the sign of every coordinate of
$\hat{\mathbf{x}}_i^*-\hat{\mathbf{y}}_k$. Hence the candidate set has size one, and
Algorithm~\ref{algo: Fixed Point Affine Functions} reduces to a single candidate
solution.
\end{proof}

\begin{ex}
In Example~\ref{ex: Example 2}, take the affine rule with $\gamma=0.2$. Then
\begin{align*}
\Gamma_1=\Gamma_2
&=
\begin{bmatrix}
0 & 0.1 & 0
\end{bmatrix}^{\top},
&
\Omega_1=\Omega_2
&=
\begin{bmatrix}
0 & 0.25 & 0
\end{bmatrix}^{\top}.
\end{align*}
Hence one may take
\begin{align*}
\underline{\Phi}_1=\underline{\Phi}_2
&=
\begin{bmatrix}
0 & 0.15 & 0
\end{bmatrix}^{\top},
&
\overline{\Phi}_1=\overline{\Phi}_2
&=
\begin{bmatrix}
0 & 0.25 & 0
\end{bmatrix}^{\top}.
\end{align*}
Theorem~\ref{thm:BEPNWF} yields
\begin{align*}
\mathbf{x}_L^*
&=
\begin{bmatrix}
0.2024 & 0.5667 & 0.1889 & 0.7439 & 0.5500 & 0.1889
\end{bmatrix}^{\top},\\
\mathbf{x}_U^*
&=
\begin{bmatrix}
0.3311 & 0.9667 & 0.3222 & 0.8848 & 0.8500 & 0.3222
\end{bmatrix}^{\top}.
\end{align*}
Since
\begin{align*}
\hat{\mathbf{y}}_1=
\begin{bmatrix}
1 & 1
\end{bmatrix}^{\top}
\end{align*}
and every entry of $\mathbf{x}_U^*$ is strictly smaller than $1$,
condition~\eqref{eq:cor2} holds. Therefore Corollary~\ref{thm:analytical} applies
and
\begin{align*}
\mathbf{x}^*
&=
\begin{bmatrix}
0.2783 & 0.8032 & 0.2677 & 0.8267 & 0.7260 & 0.2677
\end{bmatrix}^{\top}.
\end{align*}
\label{ex:Example4}
\end{ex}

\section{Network Simulation}

\label{Real Network Simulation}

We use the Teenage Friends and Lifestyle Study \cite{michell2000smoke} (also see similar works \cite{block2014multidimensional,steglich2010dynamic,snijders2017stochastic}). The data contain three waves of behavior and friendship links. We keep a connected subnetwork of $41$ women. The construction is meant to be illustrative, i.e., we do not claim a calibrated structural estimate.
    \begin{figure}[http]
\centering
\includegraphics[scale=0.35]{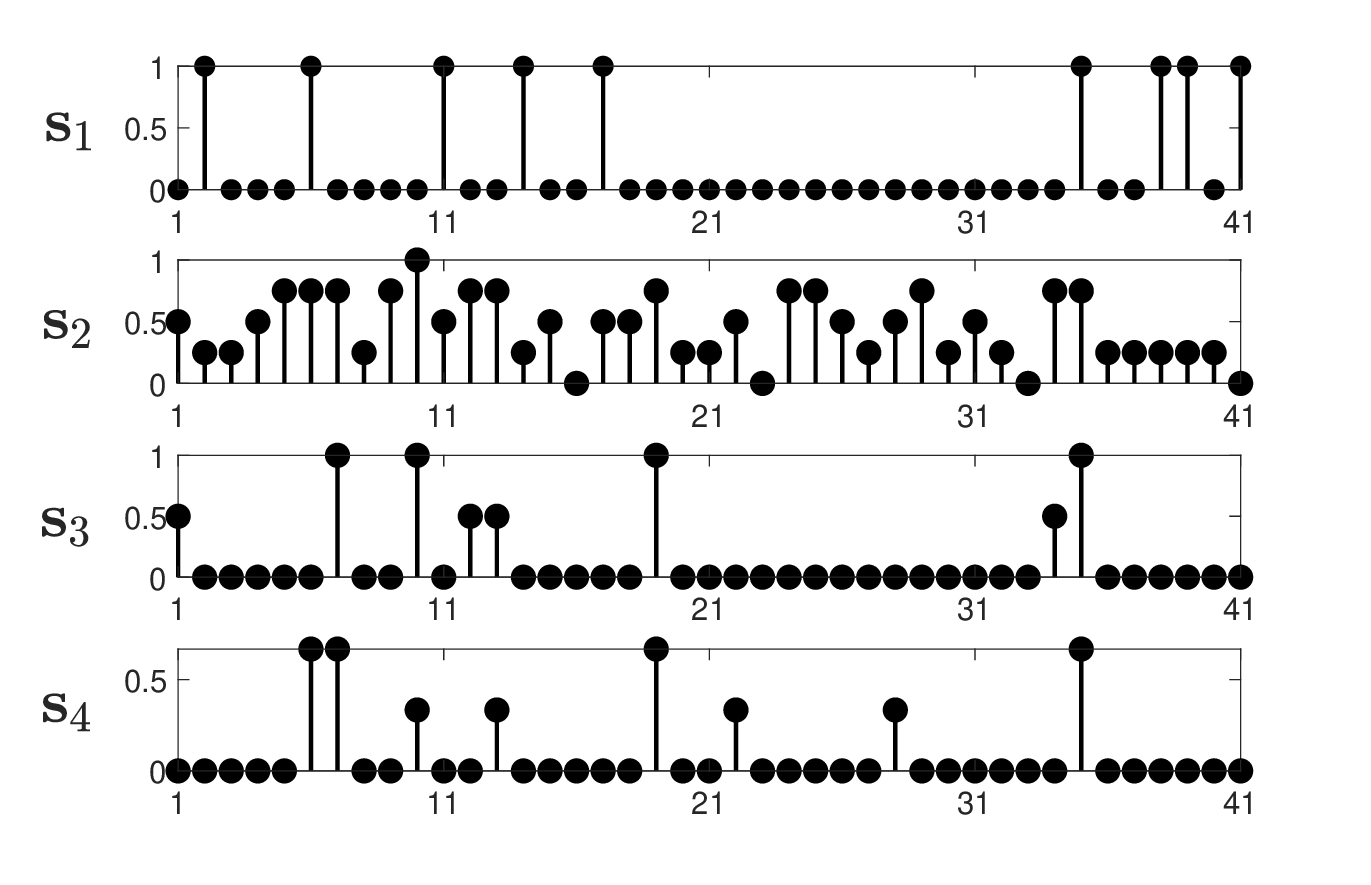}
\caption{Innate opinions in the real network.}
\label{fig: Real Network innate opinion}
\end{figure}
The four behavioral coordinates (opinions) are sport, alcohol, smoking, and cannabis. We rescale each one to $[0,1]$. For sport, $0$ means regular and $1$ means irregular. For the other three coordinates, $0$ means no use and $1$ means regular use. We denote the normalized coordinates by $\mathbf{x}_1,\dots,\mathbf{x}_4$. The normalized data from the first wave is used as the innate opinion vector $\mathbf{s}_l$, depicted in Figure \ref{fig: Real Network innate opinion}.

  Figure \ref{fig: Yearly comparison of friendship networks} depicts the friendship networks. The weighted adjacency matrices $W_1$ through $W_4$ are assumed to be identical and are obtained from a row-normalized adjacency matrix derived from Figure \ref{fig: Yearly comparison of friendship networks}.

 To determine $\boldsymbol{\alpha}_l(i)$ for the agent $\mathrm{v}_i$, we consider its network links and the shifts in its opinion throughout three years of data collection. In the given dataset, denote $\hat{\mathbf{x}}_i^p$ as the normalized opinion vector of $\mathrm{v}_i$ during the $p$-th data collection round, where $p \in \left\{1, 2, 3\right\}$, and let ${\zeta}_i$ denote the number of edges that enter the vertex $\mathrm{v}_i$. We define the empirical three-wave average $\hat{\mathbf{x}}^{\mathrm{emp}}_i \triangleq \frac{1}{3} \sum_{p = 1}^{3}{\hat{\mathbf{x}}_i^p}$,
and use the heuristic rule
    \begin{align*}
    \boldsymbol{\alpha}_l(i) &= 1-\frac{1}{1+{\frac{1}{3}}{\zeta}_i{\sqrt{{\sum_{p = 1}^{3}{\left(\hat{\mathbf{x}}_i^p(l)-\hat{\mathbf{x}}^{\mathrm{emp}}_i(l)\right)}^2}}}}.
    \end{align*}
To determine  ${\Lambda}_{l,j}$, we first compute the empirical covariance
    \begin{align*}
      \hat{\overline{\mathbf{x}}}^p = \frac{1}{n} \sum_{i = 1}^{n}{\hat{\mathbf{x}}_i^p}, \quad 
    \overline{C} = \frac{1}{3n} \sum_{p = 1}^{3} \sum_{i = 1}^{n}{(\hat{\mathbf{x}}_i^p - \hat{\overline{\mathbf{x}}}^p)(\hat{\mathbf{x}}_i^p - \hat{\overline{\mathbf{x}}}^p)^{\top}}.
    \end{align*}
    Because the model requires nonnegative cross-layer weights, we project $\overline{C}$ onto the nonnegative orthant:
\begin{align*}
\overline{C}^{+}(l,j) = \max\!\left\{\overline{C}(l,j),0\right\}.
\end{align*}
Assuming each row sum of $\overline{C}^{+}$ is positive, define the row-normalized matrix
\begin{align*}
C = \Sigma_{\overline{C}^{+}}^{-1}\overline{C}^{+},
\qquad
\Lambda_{l,j} = C(l,j)I.
\end{align*}
Then the cross-layer coefficients used in the simulation are nonnegative and satisfy $\sum_{j=1}^q \Lambda_{l,j} = I$
for every $l$. This guarantees that the cross-layer coefficients are nonnegative and satisfy
$\sum_{j=1}^q \Lambda_{l,j}=I$ for every $l$.

      $B_l\left( {{\mathbf{x}}} \right)$  is modeled by the affine function in \eqref{eq: Affine function}. Since this experiment uses a single information source, we restrict the admissible source strategy to the continuous box $
    \mathcal{Y} \triangleq [0,1]^4,
$
    which matches the normalized scale of the four behavioral variables, and we set
    \begin{align*}
        c_{l,i,1} &= C(l,:)^{\top}, \\
        {\Omega_l(i,1)} &=
        \begin{cases}
        0.9\,{n^{\omega}_i}\,\dfrac{1-{\boldsymbol{\alpha}_l(i)}}{\boldsymbol{\alpha}_l(i)},
        &
        \begin{aligned}[t]
        &\text{if } {\boldsymbol{\alpha}_l(i)} \neq 0 \\
        &\text{and } (\mathrm{u}_1,l,\mathrm{v}_i,l) \in \mathcal{E},
        \end{aligned}
        \\
        0, & \text{otherwise},
        \end{cases} \\
        {\Gamma_l(i,1)} &= {n^{\gamma}_i}{\Omega_l(i,1)}.
    \end{align*}
    Here $n_i^\omega\in\{0,1\}$ indicates whether agent $i$ is connected to the source, and $n_i^\gamma\in\{0,1\}$ indicates whether that connection is filtered by confirmation bias. Thus $(n_i^\omega,n_i^\gamma)=(0,0)$ means no source access, $(1,0)$ means source access without confirmation bias, and $(1,1)$ means source access with confirmation bias. Because $c_{l,i,1}$ is nonnegative and sums to one, every $\hat{\mathbf{y}}_1\in\mathcal{Y}$ yields $\beta_l(i,1)\in[0,1]$. Also,
\begin{align*}
0\le \Gamma_l(i,1)\le \Omega_l(i,1)\le 0.9\,\frac{1-\alpha_l(i)}{\alpha_l(i)}
\end{align*}
whenever $\alpha_l(i)\neq 0$.Hence \eqref{eq:AffineSufficient2} holds uniformly over $\mathcal{Y}$.

\begin{figure}[http]
\centering
\includegraphics[scale=0.4]{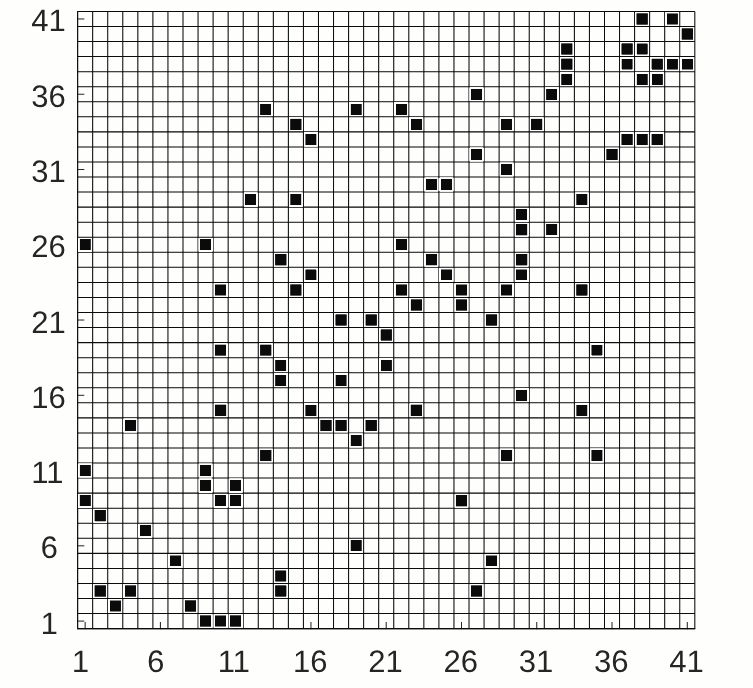}
\caption{Row-normalized friendship graph used in the simulation: $41\times 41$ matrix represents each agent in the rows and columns; a black cell at $(i,j)$ means agents $i$ and $j$ are connected.}
\label{fig: Yearly comparison of friendship networks}
\end{figure}
\subsection{Simulation}

Let $\mathbf{x}^*(\hat{\mathbf{y}}_1)$ denote the unique fixed point of \eqref{eq:CBUR} associated with an admissible source strategy $\hat{\mathbf{y}}_1\in\mathcal{Y}$, and let
\begin{align*}
\hat{\mathbf{x}}_i^*(\hat{\mathbf{y}}_1) \triangleq H_i \mathbf{x}^*(\hat{\mathbf{y}}_1), \, \,  
\hat{\overline{\mathbf{x}}}^{*}(\hat{\mathbf{y}}_1) \triangleq\frac{1}{n}\sum_{i=1}^{n}\hat{\mathbf{x}}_i^*(\hat{\mathbf{y}}_1),\,\, 
\hat{\overline{\mathbf{s}}} \triangleq \frac{1}{n}\sum_{i=1}^{n}\hat{\mathbf{s}}_i.
\end{align*}
The information source seeks a harm-reduction strategy that minimizes steady-state cannabis use, namely
\begin{align*}
J(\hat{\mathbf{y}}_1)\triangleq \hat{\overline{\mathbf{x}}}^{*}(\hat{\mathbf{y}}_1)(4),
\qquad \hat{\mathbf{y}}_1\in\mathcal{Y}.
\end{align*}
To summarize network-level changes, we use the percentage relative steady-state opinion change
\begin{align*}
\Delta_{\%}\hat{\overline{\mathbf{x}}}^{*}(l)
\triangleq
100\,\frac{\hat{\overline{\mathbf{x}}}^{*}(\hat{\mathbf{y}}_1)(l)-\hat{\overline{\mathbf{s}}}(l)}
{\hat{\overline{\mathbf{s}}}(l)}.
\end{align*}

Each objective evaluation is therefore a nested fixed-point computation: for a given $\hat{\mathbf{y}}_1\in\mathcal{Y}$, we iterate \eqref{eq:CBUR} until $\|\mathbf{x}[t+1]-\mathbf{x}[t]\|_\infty<10^{-8}$ and then compute $J(\hat{\mathbf{y}}_1)$ from the resulting steady state\footnote{Here, we do not employ Algorithm~\ref{algo: Fixed Point Affine Functions} although, in general, it provides an exact characterization and
solver for the affine fixed point but we instead use fixed-point iterations because the
steady state must be recomputed repeatedly inside an outer optimization loop.}. Because the admissible set $\mathcal{Y}$ was chosen to satisfy the uniform conditions discussed above, every candidate $\hat{\mathbf{y}}_1$ remains inside the invariant and contractive regime covered by Theorem \ref{thm:th1}.

We evaluate four scenarios: i) no information source, i.e., $n_i^{\omega}=n_i^{\gamma}=0$ for all $i$; ii) a universally accessible information source without confirmation bias, i.e., $n_i^{\omega}=1$ and $n_i^{\gamma}=0$ for all $i$; iii) a universally accessible information source with confirmation bias, i.e., $n_i^{\omega}=1$ and $n_i^{\gamma}=1$ for all $i$; and iv) source design over $\mathcal{Y}$ with confirmation bias and source access restricted to a selected subset of agents. In Scenario 4, source access is restricted to a four-agent subset. We report the numerically best subset returned by the search routine together with its optimized source vector; this should be interpreted as an illustrative computation rather than a global optimum.

Outcomes are reported in Table \ref{table:2}. We remark again that because the parameterization is heuristic, the exact percentages should be interpreted illustratively rather than a concrete scientific finding. The  results are outlined as follows:
 \begin{itemize}
  \item Without any information source, the average cannabis usage would rise by around $14\%$ from its original average.
  \item When the information source is universally available and people accept the information without confirmation bias, the optimal information vector ${\hat{\mathbf{y}}}^*_1$ is zero. This leads to a reduction in average cannabis usage by approximately $40\%$.
  \item Under the assumption that the information source is universally accessible and people accept the information with confirmation bias, it is recommended to use alcohol once or twice a year as a harm reduction strategy. This would lead to a decrease in cannabis usage by around $24\%$.
  \item With localized access to the information source and confirmation bias, cannabis usage would decrease by about $9\%$ relative to the initial average if the information source recommends avoiding cannabis but allows occasional smoking and alcohol consumption once a week. The reported 4-agent source-access subset is $\{\mathrm{v}_6,\mathrm{v}_7,\mathrm{v}_{13},\mathrm{v}_{28}\}$.
\end{itemize}

\begin{rem}
The continuous box-constrained optimization over $\mathcal{Y}$ in Scenarios 3 and 4 is implemented with MATLAB's Optimization Toolbox \cite{OptimizationToolbox}, with each candidate objective value obtained from the converged fixed point of \eqref{eq:CBUR}. 
\end{rem}
 
\begin{table}[h!]
\caption{Steady-state cannabis outcomes under four source-design scenarios.}
\centering
\footnotesize
\setlength{\tabcolsep}{4pt}
\begin{tabular}{|c|c|c|c|c|}
\hline
 scenario& 1 & 2 & 3 & 4 \\ 
\hline
${\hat{\mathbf{y}}}_1^*$
& -
& $\mathbf{0}$
& $\left[\begin{smallmatrix} 0.0 \\ 0.26\\ 0.0\\ 0.0 \end{smallmatrix}\right]$
& $\left[\begin{smallmatrix} 0.09 \\ 0.64\\ 0.57\\ 0.0 \end{smallmatrix}\right]$ \\ 
\hline
$\Delta_{\%}\hat{\overline{\mathbf{x}}}^{*}(4)$ 
& $13.5\%$
& $-39.9\%$
& $-23.63\%$
& $-8.9\%$ \\
\hline
$J(\hat{\mathbf{y}}_1)$ 
& $0.111$
& $0.059$
& $0.075$
& $0.089$ \\
\hline
\end{tabular}
\label{table:2}
\end{table}

These findings suggest that design conclusions can be sensitive to whether
confirmation bias is included in the model. Under the present illustrative
parameterization, the optimizer selects different source vectors in the no-bias and
confirmation-bias settings. We therefore interpret the results qualitatively, as evidence that confirmation bias can change the direction of the recommended intervention. Also note that the best strategy identified under no-bias conditions advocates for avoiding all risky behaviors to minimize high-risk actions. Nonetheless, this advice does not conform to \textit{Principles of Harm Reduction}, which argue that fully prohibitive strategies may not be as effective as those promoting moderation or substitution \cite{marlatt2011harm}. On the other hand, strategies using models that include  bias are more pragmatic. For instance, scenario 4 demonstrates that better outcomes are observed when the recommendations involve avoiding cannabis but allow occasional smoking and weekly alcohol consumption. Another illustrative result is that pragmatic strategies depend on the accessibility of information sources. For example, in a universally accessible scenario, the best strategy may prohibit risky behaviors more effectively compared to localized information sources.

\section{Conclusion}
\label{Conclusion}

We studied a topic-layered multidimensional opinion-dynamics model with confirmation bias. The model combines FJ-type peer interaction with state-dependent source influence, where source weights depend on agent--source opinion mismatch. It can be viewed as a restricted node-aligned multiplex system in which the same agents and sources appear in each topic layer, intra-layer social interaction is encoded by $W_l$, and inter-layer coupling is encoded by $\Lambda_{l,j}$.

For general Lipschitz source-influence functions, we established sufficient conditions for contraction and convergence to a unique fixed point. For affine confirmation-bias functions, we showed that the fixed point can be computed through a finite sign-consistency search and, in a special regime, in closed form. For bounded nonlinear source-influence functions, we derived explicit lower and upper bounds on the fixed point.

The examples and the real-network study show that significance of multidimensional coupling  and that source-design conclusions can change qualitatively when confirmation bias is ignored. Because the network parameterization is heuristic, these conclusions should be interpreted qualitatively.

It is important to note the limitations of the multilayer modeling in the present paper: it uses a node-aligned cross-topic coupling, i.e.,  each $\Lambda_{l,j}$ is diagonal, so agent $i$ in layer $l$ reweights only its own layer-$j$ social aggregate $(W_j\mathbf{x}_j)(i)$. A more general multilayer model, in the sense of \cite{kivela2014multilayer,boccaletti2014structure}, would replace $\Lambda_{l,j}$ by full nonnegative matrices $M_{l,j}\in\mathbb{R}_+^{n\times n}$. Under nonnegativity, the current normalization $\sum_{j=1}^q \Lambda_{l,j}=I$ forces all off-diagonal entries to vanish, and therefore cannot represent genuine interlayer edges between distinct agent-layer nodes. A true interlayer generalization would instead require a weaker condition such as
\begin{align*}
\sum_{j=1}^q M_{l,j}\mathbf{1}_n=\mathbf{1}_n,
\qquad l\in[q],
\end{align*}
or, equivalently, given row-stochastic $W_j$, that the supra-social matrix
\begin{align*}
\mathcal{W}\triangleq [M_{l,j}W_j]_{l,j=1}^q
\end{align*}
be row-stochastic. In that setting, topic-$j$ social evidence formed around agent $r$ could affect the update of topic $l$ for a different agent $i$, allowing cross-agent cross-topic spillovers, coordinated exposure, and latent group-level coupling that are excluded by the present diagonal restriction. While the contraction analysis is likely to extend by replacing $\|AW\|_\infty$ with $\|A\mathcal{W}\|_\infty$, however the behavioral interpretation, identifiability, and decoupling structure would require a separate analysis. It would also be interesting to combine such fully multilayer coupling with time-varying multiplex interaction patterns as in \cite{PROSKURNIKOV201711896, shiu2026coordination, abedinzadeh2025tvfj}. These  as well as other natural extensions such as hierarchical leader-mediated communication \cite{wang2026molena}, and strategic source design \cite{mao2021competitive} are left as future work. 
\appendices

\appendices

\section{Proof of Theorem~\ref{thm:th1}}
\label{Proof of Theorem 1}
Let $f$ denote the map in \eqref{eq:CBUR}. We first show that
$f(\mathbb{I}^{nq})\subseteq \mathbb{I}^{nq}$. For $l\in[q]$ and $i\in[n]$,
\eqref{eq:ABUR} gives
\begin{align*}
f(\mathbf{x})((l-1)n+i)
&=
\Bigl(1-\alpha_l(i)\bigl(1+\sum_{k=1}^m b_{l,i,k}(\mathbf{x})\bigr)\Bigr)
\mathbf{s}_l(i) \\
&\quad+
\alpha_l(i)\sum_{j=1}^q \Lambda_{l,j}(i,i)
\sum_{a=1}^n W_j(i,a)\,\mathbf{x}_j(a) \\
&\quad+
\alpha_l(i)\sum_{k=1}^m b_{l,i,k}(\mathbf{x})\,\mathbf{y}_l(k).
\end{align*}
By Assumption~\ref{asm:Positivity}, all source weights are nonnegative. By
Assumption~\ref{asm:Wrow}, each $W_j$ is row-stochastic, and since
$\sum_{j=1}^q\Lambda_{l,j}=I$, we have
\begin{align*}
\sum_{j=1}^q \Lambda_{l,j}(i,i)\sum_{a=1}^n W_j(i,a)=1.
\end{align*}
By Assumption~\ref{asm: ConvexCondition2},
\begin{align*}
\alpha_l(i)\Bigl(1+\sum_{k=1}^m b_{l,i,k}(\mathbf{x})\Bigr)\le 1,
\end{align*}
so the coefficient of $\mathbf{s}_l(i)$ is nonnegative. Therefore each coordinate
of $f(\mathbf{x})$ is a convex combination of values in $[0,1]$. Hence
$f(\mathbf{x})\in\mathbb{I}^{nq}$.

We next prove contraction. Fix $\mathbf{x},\mathbf{z}\in\mathbb{I}^{nq}$.
For each $l\in[q]$ and $i\in[n]$,
\begin{align*}
&\bigl|(f(\mathbf{x})-f(\mathbf{z}))((l-1)n+i)\bigr| \\
&\le
\alpha_l(i)
\left|
\sum_{j=1}^q \Lambda_{l,j}(i,i)
\sum_{a=1}^n W_j(i,a)
\bigl(\mathbf{x}_j(a)-\mathbf{z}_j(a)\bigr)
\right| \\
&\quad+
\alpha_l(i)\sum_{k=1}^m
\bigl|b_{l,i,k}(\mathbf{x})-b_{l,i,k}(\mathbf{z})\bigr|
\,\bigl|\mathbf{y}_l(k)-\mathbf{s}_l(i)\bigr|.
\end{align*}
Using row-stochasticity of $W_j$, the identity
$\sum_{j=1}^q\Lambda_{l,j}(i,i)=1$, and Assumption~\ref{asm: LC}, we obtain
\begin{align*}
&\bigl|(f(\mathbf{x})-f(\mathbf{z}))((l-1)n+i)\bigr| \\
&\le
\alpha_l(i)
\sum_{j=1}^q \Lambda_{l,j}(i,i)
\sum_{a=1}^n W_j(i,a)
\|\mathbf{x}-\mathbf{z}\|_\infty \\
&\quad+
\alpha_l(i)\sum_{k=1}^m \mu_{l,i,k}
\,\bigl|\mathbf{y}_l(k)-\mathbf{s}_l(i)\bigr|
\|\mathbf{x}-\mathbf{z}\|_\infty \\
&=
\alpha_l(i)
\left(
1+
\sum_{k=1}^m \mu_{l,i,k}
\,\bigl|\mathbf{y}_l(k)-\mathbf{s}_l(i)\bigr|
\right)
\|\mathbf{x}-\mathbf{z}\|_\infty.
\end{align*}
Define
\begin{align*}
\kappa_{l,i}
\triangleq
\alpha_l(i)
\left(
1+
\sum_{k=1}^m \mu_{l,i,k}
\,\bigl|\mathbf{y}_l(k)-\mathbf{s}_l(i)\bigr|
\right).
\end{align*}
Then
\begin{align*}
\bigl|(f(\mathbf{x})-f(\mathbf{z}))((l-1)n+i)\bigr|
\le
\kappa_{l,i}\,\|\mathbf{x}-\mathbf{z}\|_\infty.
\end{align*}
Taking the maximum over $l$ and $i$ gives
\begin{align*}
\|f(\mathbf{x})-f(\mathbf{z})\|_\infty
\le
\kappa\,\|\mathbf{x}-\mathbf{z}\|_\infty,
\end{align*}
where
\begin{align*}
\kappa \triangleq \max_{l\in[q],\,i\in[n]} \kappa_{l,i}.
\end{align*}
By Assumption~\ref{asm: CC}, we have $\kappa<1$. Therefore $f$ is a contraction on
$\mathbb{I}^{nq}$ under $\|\cdot\|_\infty$. Since $\mathbb{I}^{nq}$ is complete
under $\|\cdot\|_\infty$, the Banach fixed-point theorem (cf. \cite{rudin1991functional}) yields a unique fixed
point and convergence from every initial condition.

\section{Proof of Theorem~\ref{thm:th2}}
\label{Proof of Theorem 2}
Under \eqref{eq: Affine function}, each
\begin{align*}
b_{l,i,k}(\mathbf{x})
=
\Omega_l(i,k)-\Gamma_l(i,k)\beta_l(i,k),
\end{align*}
where
\begin{align*}
\beta_l(i,k)=c_{l,i,k}^{\top}|H_i\mathbf{x}-T_k\mathbf{y}|.
\end{align*}
Because $H_i$ and $T_k$ are selection matrices, $H_i\mathbf{x},T_k\mathbf{y}
\in\mathbb{I}^q$. Since $c_{l,i,k}\ge 0$ and $\mathbf{1}^{\top}c_{l,i,k}=1$,
we have
\begin{align*}
0\le \beta_l(i,k)\le 1
\end{align*}
for all $\mathbf{x}\in\mathbb{I}^{nq}$. Hence \eqref{eq:AffineSufficient1}
implies Assumptions~\ref{asm:Positivity} and \ref{asm: ConvexCondition2}. Moreover,
for $\mathbf{x},\mathbf{z}\in\mathbb{I}^{nq}$,
\begin{align*}
&\bigl|c_{l,i,k}^{\top}|H_i\mathbf{x}-T_k\mathbf{y}|
-c_{l,i,k}^{\top}|H_i\mathbf{z}-T_k\mathbf{y}|\bigr| \\
&\le c_{l,i,k}^{\top}|H_i(\mathbf{x}-\mathbf{z})| \\
&\le \|\mathbf{x}-\mathbf{z}\|_\infty,
\end{align*}
so $b_{l,i,k}$ is Lipschitz with constant $\Gamma_l(i,k)$. Together with
Assumption~\ref{asm:Wrow} and \eqref{eq:kappa-aff}, Theorem~\ref{thm:th1} applies.
Therefore \eqref{eq:CBUR} has a unique fixed point $\mathbf{x}^*$.

Fix a candidate sign pattern
\begin{align*}
\boldsymbol{\theta}=\{\boldsymbol{\theta}_{i,k}\},
\qquad
\boldsymbol{\theta}_{i,k}\in\{-1,0,1\}^q,
\end{align*}
and define
\begin{align*}
\Theta_{i,k}\triangleq \diag(\boldsymbol{\theta}_{i,k}).
\end{align*}
Whenever
\begin{align*}
\sign(H_i\mathbf{x}-T_k\mathbf{y})=\boldsymbol{\theta}_{i,k},
\end{align*}
we have
\begin{align*}
|H_i\mathbf{x}-T_k\mathbf{y}|
=
\Theta_{i,k}(H_i\mathbf{x}-T_k\mathbf{y}),
\end{align*}
and therefore
\begin{align*}
\beta_l(i,k)
&=
 c_{l,i,k}^{\top}|H_i\mathbf{x}-T_k\mathbf{y}| \\
&=
 c_{l,i,k}^{\top}\Theta_{i,k}H_i\mathbf{x}
-
 c_{l,i,k}^{\top}\Theta_{i,k}T_k\mathbf{y}.
\end{align*}
Thus, once the sign pattern is fixed, the dynamics become affine.

For each $l\in[q]$ and $i\in[n]$, define
\begin{align*}
r_{l,i}^{x}(\boldsymbol{\theta})
&\triangleq
\sum_{k=1}^{m}
\Gamma_l(i,k)\bigl(\mathbf{s}_l(i)-\mathbf{y}_l(k)\bigr)
 c_{l,i,k}^{\top}\Theta_{i,k}H_i, \\
r_{l,i}^{y}(\boldsymbol{\theta})
&\triangleq
\sum_{k=1}^{m}
\Gamma_l(i,k)\bigl(\mathbf{y}_l(k)-\mathbf{s}_l(i)\bigr)
 c_{l,i,k}^{\top}\Theta_{i,k}T_k.
\end{align*}
Let $R_l^x(\boldsymbol{\theta})\in\mathbb{R}^{n\times nq}$ and
$R_l^y(\boldsymbol{\theta})\in\mathbb{R}^{n\times mq}$ collect these rows, and
stack them as
\begin{align*}
R^x(\boldsymbol{\theta})
&\triangleq
\begin{bmatrix}
R_1^x(\boldsymbol{\theta})\\
\vdots\\
R_q^x(\boldsymbol{\theta})
\end{bmatrix},
&
R^y(\boldsymbol{\theta})
&\triangleq
\begin{bmatrix}
R_1^y(\boldsymbol{\theta})\\
\vdots\\
R_q^y(\boldsymbol{\theta})
\end{bmatrix}.
\end{align*}
Also define
\begin{align*}
\Omega
&\triangleq \blkdiag(\Omega_1,\dots,\Omega_q), \\
\Sigma_{\Omega}
&\triangleq \blkdiag(\Sigma_{\Omega_1},\dots,\Sigma_{\Omega_q}), \\
A_a
&\triangleq I-A-A\Sigma_{\Omega}, \\
W_a(\boldsymbol{\theta})
&\triangleq A\bigl(W+R^x(\boldsymbol{\theta})\bigr), \\
B_a(\boldsymbol{\theta})
&\triangleq A\bigl(\Omega+R^y(\boldsymbol{\theta})\bigr).
\end{align*}
Then the stacked dynamics reduce to
\begin{align*}
\mathbf{x}[t+1]
=
A_a\mathbf{s}
+
B_a(\boldsymbol{\theta})\mathbf{y}
+
W_a(\boldsymbol{\theta})\mathbf{x}[t].
\end{align*}

It remains to show that $I-W_a(\boldsymbol{\theta})$ is nonsingular for every
candidate pattern. For fixed $l$ and $i$, the row vector
$c_{l,i,k}^{\top}\Theta_{i,k}H_i$ has absolute row sum at most
$\mathbf{1}^{\top}c_{l,i,k}=1$. Hence the $(l,i)$th row of
$AR^x(\boldsymbol{\theta})$ has absolute row sum at most
\begin{align*}
\alpha_l(i)\sum_{k=1}^{m}
\Gamma_l(i,k)\,|\mathbf{y}_l(k)-\mathbf{s}_l(i)|.
\end{align*}
Since each $W_j$ is row-stochastic and $\sum_{j=1}^q\Lambda_{l,j}=I$, the
$((l-1)n+i)$th row sum of $AW$ equals $\alpha_l(i)$. Therefore the absolute row
sum of the $((l-1)n+i)$th row of $W_a(\boldsymbol{\theta})$ is at most
\begin{align*}
\kappa_{l,i}^{\mathrm{aff}}
\triangleq
\alpha_l(i)\left(
1+
\sum_{k=1}^{m}\Gamma_l(i,k)
\,|\mathbf{y}_l(k)-\mathbf{s}_l(i)|
\right).
\end{align*}
By \eqref{eq:kappa-aff},
\begin{align*}
\kappa_{\mathrm{aff}}
\triangleq
\max_{l\in[q],\,i\in[n]} \kappa_{l,i}^{\mathrm{aff}} < 1.
\end{align*}
Hence
\begin{align*}
\|W_a(\boldsymbol{\theta})\|_\infty \le \kappa_{\mathrm{aff}} < 1,
\end{align*}
so $I-W_a(\boldsymbol{\theta})$ is nonsingular.

For a fixed pattern $\boldsymbol{\theta}$, the affine system has the unique
equilibrium
\begin{align*}
\mathbf{x}_{\boldsymbol{\theta}}
=
\bigl(I-W_a(\boldsymbol{\theta})\bigr)^{-1}
\bigl(A_a\mathbf{s}+B_a(\boldsymbol{\theta})\mathbf{y}\bigr).
\end{align*}
Each coordinate of each $\boldsymbol{\theta}_{i,k}$ takes values in
$\{-1,0,1\}$, so the candidate set is finite. The true fixed point $\mathbf{x}^*$
induces one such sign pattern, hence at least one candidate is consistent.
Conversely, if a candidate satisfies
\begin{align*}
\sign\bigl(H_i\mathbf{x}_{\boldsymbol{\theta}}-T_k\mathbf{y}\bigr)
=
\boldsymbol{\theta}_{i,k}
\end{align*}
for all relevant $i$ and $k$, then $\mathbf{x}_{\boldsymbol{\theta}}$ satisfies the
original nonlinear fixed-point equation. By uniqueness of the fixed point, exactly one
candidate passes the consistency check. Therefore
Algorithm~\ref{algo: Fixed Point Affine Functions} returns $\mathbf{x}^*$ after
finitely many steps.

\bibliographystyle{IEEEtran}
\bibliography{ref}

\end{document}